\documentclass[journal,comsoc]{IEEEtran}

\usepackage{xcolor}
\usepackage{textcomp}
\usepackage{algorithmic}
\usepackage{multirow}
\usepackage{graphicx}
\usepackage{mathtools}  
\usepackage{amsmath}
\usepackage{amsthm}
\usepackage{amssymb}
\usepackage{tabulary}
\usepackage{booktabs}

\usepackage[ruled,linesnumbered]{algorithm2e}

\usepackage{comment}
\usepackage{subfig}
\usepackage{float}
\usepackage{adjustbox}
\emergencystretch=\maxdimen
\usepackage[cmintegrals]{newtxmath}
\DeclareMathOperator{\E}{\mathbb{E}}

\newtheorem{theorem}{Theorem}

\newtheorem{dfn}{Definition}

\usepackage{epstopdf}
\AppendGraphicsExtensions{.tif}

\begin{document}
\title{Optimization of End-to-End AoI in Edge-Enabled Vehicular Fog Systems: A Dueling-DQN Approach}

\author{Seifu~Birhanu~Tadele, 
        Binayak~Kar,~\IEEEmembership{Member,~IEEE}, 
        Frezer~Guteta~Wakgra, and 
        Asif~Uddin~Khan,~\IEEEmembership{Member,~IEEE}

\thanks{This work was supported by the National Science and Technology Council (NSTC), Taiwan, under Grant 112-2221-E-011-057 -}
\thanks{S. B. Tadele, B. Kar, and F. G. Wakgra are with the Department of Computer Science and Information Engineering, National Taiwan University of Science and Technology, Taipei 106, Taiwan (E-mail: seife.brhan@gmail.com, bkar@mail.ntust.edu.tw, haftiyam@gmail.com.)}
\thanks{A. U. Khan is with the School of Computer Engineering, KIIT Deemed to be University, Bhubaneswar, Odisha, India. (E-mail: asif.khanfcs@kiit.ac.in)}
}

\maketitle

\begin{abstract}
In real-time status update services for the Internet of Things (IoT), the timely dissemination of information requiring timely updates is crucial to maintaining its relevance. Failing to keep up with these updates results in outdated information. The age of information (AoI) serves as a metric to quantify the freshness of information. The Existing works to optimize AoI primarily focus on the transmission time from the information source to the monitor, neglecting the transmission time from the monitor to the destination. This oversight significantly impacts information freshness and subsequently affects decision-making accuracy. To address this gap, we designed an edge-enabled vehicular fog system to lighten the computational burden on IoT devices. We examined how information transmission and request-response times influence end-to-end AoI.
As a solution, we proposed Dueling-Deep Queue Network (dueling-DQN), a deep reinforcement learning (DRL)-based algorithm, and compared its performance with DQN policy and analytical results. Our simulation results demonstrate that the proposed dueling-DQN algorithm outperforms both DQN and analytical methods, highlighting its effectiveness in improving real-time system information freshness. Considering the complete end-to-end transmission process, our optimization approach can improve decision-making performance and overall system efficiency.
 
\end{abstract}


\section{Introduction}
\IEEEPARstart{T}{he} integration of wireless communications with IoT devices has been driven by the increasing demand for real-time status updates \cite{chen2022minimizingRev}. Real-time status updates have diverse applications, ranging from optimizing energy consumption in smart homes \cite{Menouar2017} 
to predicting and managing forest fires by monitoring temperature and humidity \cite{bisquert2012application}. 
Additionally, they play a crucial role in assisting drivers by providing data on a vehicle’s speed and acceleration in intelligent transportation systems \cite{xiong2012intelligent}. To measure the timeliness of this information, a metric called AoI has been introduced \cite{Kaul2012}. It measures the duration between the latest received information and its generation time. It is of great importance in IoT applications where the timeliness of information is crucial, e.g., in monitoring the status of a system or estimating a Markov process \cite{chen2022age}. Unlike conventional measurements like throughput and latency, AoI considers the generation time and delay in the transmission of information \cite{Xie2022}.

IoT systems are crucial in processing and transmitting real-time status updates \cite{Roy2019, hu2021status_9}. 
However, the constrained processing capabilities of IoT devices may present challenges in handling real-time status updates for computationally intensive packets. This challenge can be mitigated by offloading real-time data to edge computing \cite{kar2023cost, wakgra2024multi} and vehicular fog computing (VFC), ensuring rapid processing and the preservation of information freshness. 
Edge computing, rooted in the European Telecommunication Standards Institute's cloud virtualization concept, positions servers near cellular base stations to process data closer to users, ensuring faster response times \cite{hu2015mobile}.
VFC optimizes vehicular communication and computational resources, leveraging vehicles as the underlying infrastructure \cite{kar2021qos}. Vehicles can become intelligent entities through an edge server, interacting with their surroundings (V2I) \cite{Sorkhoh2020}. VFC improves the Intelligent Traffic System (ITS) by processing real-time traffic information, providing accident alerts, facilitating path navigation, and more \cite{keshari2022survey}.

Several studies have been conducted on the optimization of AoI based on the transmission time of packets from the source to the monitor \cite{chen2023improving_14, Chao2020, Qiaobin2020}. However, as illustrated in Figure \ref{fig: Pmagt}, the time taken to transmit processed packets from the monitor to the destination has been overlooked. The transmission time of packets from the monitor to the destination significantly influences the freshness of information observation, thereby impacting the system’s performance in making accurate decisions. When outdated information reaches the intended destinations, it unavoidably degrades the quality and reliability of decisions drawn from it, potentially compromising safety and security. Additionally, the limited processing capabilities of IoT devices for computationally intensive packet processing can lead to the distribution of stolen information to the destination.

Let’s consider an example in intelligent transportation systems \cite{xiong2012intelligent}, where the aim is to keep drivers well-informed about their environment's status, including updates on traffic congestion, accident alerts, emergency message delivery, and earthquake alerts. The information provided by these services frequently changes over time. Consequently, the information's freshness and timely updates \cite{Yin2017, deng2022information_18} 
and efficient dissemination of this information are crucial to ensuring its delivery in a fresh state. Otherwise, outdated information will reduce its usefulness, particularly in dynamically changing environments.

We designed an edge-enabled vehicular fog system, as shown in Figure \ref{fig: Pmagt}, to reduce the computational load of collecting packets from IoT devices. The system offloads packets to an edge server and then to a vehicular fog node via wireless communication. To address this, we formulated an optimization problem for average AoI and proposed a DRL-based dueling-DQN algorithm to solve it \cite{wang2016dueling}. The key technical contributions are summarized below:

\begin{enumerate}
\item Designed edge-enabled vehicular environment, formulated and analyzed the impact of packet transmission time and request-response time on optimizing average end-to-end AoI.
\item We proposed a DRL-based dueling-DQN algorithm to solve the issue.
\item We validated the analytical results by implementing the proposed approach and constructing Neural Networks through simulation experiments using TensorFlow \cite{Abadi2016}.
\item Furthermore, we conducted a performance comparison of the proposed algorithm against both the current DQN algorithm and the analytical results.
\end{enumerate}

The rest of the paper is organized as follows. Section \ref{sec:Related-Works} provides a review of related works. Section \ref{System_Model} elucidates the system models and analyses of AoI. In Section \ref{sec:Problem_Formulation}, the paper details the problem formulations and introduces the proposed dueling-DQN solution. Section \ref{Performance} delves into the simulation and results, while Section \ref{Conclusion} concludes this work.

\section{Related Works}
\label{sec:Related-Works}
This section presents related works on AoI optimization. As described in the next section, we have categorized these works into three groups, which are explained in the next section and summarized in Table \ref{tab:AoI_servey}.

\subsubsection{AoI Optimization using Queue Models} 
Yates et al.~\cite{yates2018age} examine a straightforward queue with several sources sending updates to a monitor. To optimize AoI, they consider two types of last-come-first-served queueing systems and first-come-first-served and used stochastic hybrid systems.
Chao et al. \cite{Chao2020} developed an IoT system incorporating multiple sensors to assess information freshness, optimizing peak AoI (PAoI). Elmagid et al. \cite{Elmagid2022} considered a multi-source updating system that uses SHS for several queueing disciplines and energy harvesting (EH) to deliver status updates about multiple sources to a monitor.

\subsubsection{AoI Optimization based on Edge-enabled Models} Qiaobin et al. \cite{Qiaobin2020} employed a zero-wait policy and used MEC to examine the AoI for computation-intensive messages in their research. Ying et al. \cite{Ying2022} have researched a wireless sensor node that uses RF energy from an energy transmitter to power the transfer of compute tasks to a MEC server. The zero-wait policy is implemented when the sensor node submits the computing task to the MEC server after the capacitor has fully charged. Zipeng et al. \cite{Zipeng2022} discussed a real-time information dissemination approach over vehicular social networks (VSNs) with an AoI-centric perspective. They utilize scale-free network theory to model social interactions among autonomous vehicles (AVs).

 \begin{table*}[t]
   \centering
   \caption{Survey on AoI of optimization in different scenarios}
    \label{tab:AoI_servey}
   \resizebox{\textwidth}{!}{%
     \begin{tabular}{|l|p{3.425em}|p{7.615em}|p{3.075em}|p{2.81em}|l|p{4.73em}|p{12.77em}|l|p{11.075em}|}
     \hline
     \multicolumn{1}{|c|}{\multirow{2}[4]{*}{\textbf{Categories }}} & \multirow{2}[4]{*}{\textbf{Papers}} & \multirow{2}[4]{*}{\textbf{Target Network }} & \multicolumn{2}{p{5.885em}|}{\textbf{Source }} & \multicolumn{2}{p{8.27em}|}{\textbf{Estimation of AoI}} & \multirow{2}[4]{*}{\textbf{Objectives}} & \multicolumn{1}{c|}{\multirow{2}[4]{*}{\textbf{Constraint}}} & \multirow{2}[4]{*}{\textbf{Approaches}} \\
\cline{4-7}           & \multicolumn{1}{c|}{} & \multicolumn{1}{c|}{} & \textbf{Single} & \textbf{Multi } & \multicolumn{1}{p{3.54em}|}{\textbf{Source to Monitor }} & \textbf{Monitor to  Destinations} & \multicolumn{1}{c|}{} &       & \multicolumn{1}{c|}{} \\
     \hline
     \multicolumn{1}{|l|}{\multirow{3}[6]{*}{AoI Optimization based on Queueing Models}} & \cite{Roy2019}   & Simple Queueing Model & \multirow{3}[6]{*}{$\times$} & \multirow{3}[6]{*}{$\checkmark$} & \multicolumn{1}{l|}{\multirow{12}[24]{*}{\textbf{\checkmark}}} & \multirow{11}[22]{*}{$\times$} & Average AoI  & \multicolumn{1}{l|}{\multirow{2}[4]{*}{Queue Capacity}} & Stochastic Hybrid System (SHS) \\
\cline{2-3}\cline{8-8}\cline{10-10}           & \cite{Chao2020}  & Tandem Queueing Model  & \multicolumn{1}{l|}{} & \multicolumn{1}{l|}{} &       & \multicolumn{1}{l|}{} & Average PAoI &       & Min-max Optimization \\
\cline{2-3}\cline{8-10}           & \cite{Elmagid2022}  & General Queueing Model & \multicolumn{1}{l|}{} & \multicolumn{1}{l|}{} &       & \multicolumn{1}{l|}{} & Average AoI & \multicolumn{1}{p{6.615em}|}{Energy} & SHS \\
\cline{1-5}\cline{8-10}     \multicolumn{1}{|l|}{\multirow{3}[6]{*}{AoI Optimization based on Edge-enabled Models}} & 

\cite{Qiaobin2020}  & Mobile Edge Computing  & \multirow{2}[4]{*}{\textbf{\checkmark}} & \multirow{2}[4]{*}{$\times$} &       & \multicolumn{1}{l|}{} & Average AoI  & \multicolumn{1}{p{6.615em}|}{Resources} & Zero-wait policy \\
\cline{2-3}\cline{8-10}           & 

\cite{Ying2022}  & IoT Device & \multicolumn{1}{l|}{} & \multicolumn{1}{l|}{} &       & \multicolumn{1}{l|}{} & Average AoI & \multicolumn{1}{p{6.615em}|}{Energy} & Generating set search-based algorithm \\
\cline{2-5}\cline{8-10}           & 

\cite{Zipeng2022}  & Vehicular Network & $\times$ & \multirow{3}[6]{*}{\textbf{\checkmark}} &       & \multicolumn{1}{l|}{} & Average Peak Network AoI (NAoI)  & \multicolumn{1}{p{6.615em}|}{Transmission Capacity} & Parametric Optimization Schematic \\
\cline{1-4}\cline{8-10}     \multicolumn{1}{|l|}{\multirow{6}[12]{*}{AoI Optimization using RL Approach}} & 

\cite{Xie2022}  & IoT Device & \textbf{\checkmark} & \multicolumn{1}{l|}{} &       & \multicolumn{1}{l|}{} & Average weighted sum of AoI and energy  & \multicolumn{1}{p{6.615em}|}{Energy} & DRL \\
\cline{2-4}\cline{8-10}           & 

\cite{Zoubeir2022}  & Vehicle to Everything  & $\times$ & \multicolumn{1}{l|}{} &       & \multicolumn{1}{l|}{} & Average AoI & \multicolumn{1}{l|}{\multirow{5}[10]{*}{Resources}} & DRL(DDPG) \\
\cline{2-5}\cline{8-8}\cline{10-10}           & 

\cite{Sihua2022}  & IoT Device & \multirow{3}[6]{*}{\textbf{\checkmark}} & \multirow{3}[6]{*}{$\times$} &       & \multicolumn{1}{l|}{} & Average weighted sum of AoI and energy &       & Distributed RL \\
\cline{2-3}\cline{8-8}\cline{10-10}           & 

\cite{Lingshan2022}  & Edge Computing & \multicolumn{1}{l|}{} & \multicolumn{1}{l|}{} &       & \multicolumn{1}{l|}{} & Average AoI  &       & DQN \\
\cline{2-3}\cline{8-8}\cline{10-10}           & 

\cite{Xianfu2020}  & Vehicle Network & \multicolumn{1}{l|}{} & \multicolumn{1}{l|}{} &       & \multicolumn{1}{l|}{} & AoI-aware radio resource allocation &       & DRQN \\
\cline{2-5}\cline{7-8}\cline{10-10}           & 

Ours  & Edge and Vehicular Fog & $\times$ & \textbf{\checkmark} &       & \textbf{\checkmark} & Average end-to-end AoI &       & RL(dueling-DQN) \\
     \hline
     \end{tabular}%
      }
 \end{table*}

\subsubsection{AoI Optimization using RL Approach} 
Xie et al. \cite{Xie2022} explored an IoT system comprising devices endowed with computing capabilities. These devices can update their status, and this update can either be computed directly by the monitor or by the device itself before being transmitted to the destination. The authors aimed to minimize the average weighted sum of AoI and energy consumption. To achieve this, they employed DRL based on offloading and scheduling criteria. 
Zoubeir et al. \cite{Zoubeir2022} employed the deep deterministic policy gradient (DDPG) based DRL algorithm to reduce the average end-to-end AoI within 5G vehicle networks.
Sihua et al. \cite{Sihua2022} introduced an innovative distributed RL method for sampling policy optimization to minimize the combined weighted sum of AoI and total energy consumption within the examined IoT system, as referenced in their research. Lingshan et al. \cite{Lingshan2022} devised a novel approach for an Unmanned Aerial Vehicle (UAV)-assisted wireless network with Radio Frequency (RF) Wireless Power Transfer (WPT) using DRL. They formulated the problem as a Markov process with extensive state spaces and introduced a DQN-based strategy to discover a near-optimal solution for it. Xianfu et al. \cite{Xianfu2020} analyzed the management of radio resources in a vehicle-to-vehicle network within a Manhattan grid. Their work involves an inventive algorithm utilizing DRL techniques and long short-term memory, enabling a proactive approach within the network.

The above-mentioned works focus on optimizing the AoI of data packet transmission time from source to monitor. However, the time to transmit the results from the monitor to the destination has been ignored. This presents a challenge in addressing the needs of intelligent applications for real-time processing of packets, making it difficult to meet system performance requirements.

\section{System Model and Analysis of AoI}
\label{System_Model}
In this section, we discussed the edge-enabled vehicular fog systems model and the analysis of AoI models. Table~\ref{tab:notation} presents the notations used in AoI analysis and problem formulation.

\subsection{Edge-enabled Vehicular Fog Model}
We considered a real-time status update in edge-enabled vehicular fog systems as illustrated in Figure \ref{fig: Pmagt}. The system consists of a set of IoT devices (sources), denoted as $I_{j}$, an edge server (monitor) denoted as ${A}$, and vehicular fog (destination) denoted as ${V}$. The packets from the source are generated based on a Poisson process with a parameter $\lambda_{j}$ and served according to independent and identically distributed (i.i.d.) service times with rates $\mu^{\text{A}}$ and $\mu^{\text{$V$}}$ where $\mu^{\text{A}}$ and $\mu^{\text{$V$}}$,  representing the computational capacities of the edge server and vehicular fog, respectively. The server utilization offered by the $j$th source is $\rho_{j}=\frac{\lambda_{j}}{\mu^{\text{$A$}}}$, and the total server utilization is $\rho=\sum_{i=1}^{J} \rho_{j}$. Thus, the updates of the $j$th source compete with the aggregate of the other sources, with $\rho_{-j}=\rho-\rho_{i}=\sum_{j\neq i}^{J} \rho_{i}$, where $i\in I_j$. The edge server receives updates from the IoT devices and keeps track of the system's state. The vehicular fog can connect to the edge server via a wireless communication channel and request the latest information. The updated information from the edge server is offloaded to the vehicular fog via the wireless communication channel. Using the metric Age of AoI, we can assess how recent the status information is at the destination (target).

\begin{table}[!t]
  \centering
  \caption{Description of important notations}
  \label{tab:notation}
    \begin{tabular}{|p{2.25cm}|p{5.70cm}|}
    \hline
    \textbf{Notations} & \textbf{Descriptions} \\
    \hline
    \text{$A$} & Edge server  \\
    \hline
    \text{$V$} & Vehicular-fog \\
    \hline
    $I_j$  & Set of IoT devices, ${j}=\{1, 2, 3, \cdots, J\}$ \\
    \hline
    $d_{j}^{\text{IA}}$, $d^{\text{AV}}$  & Distance between $j$th IoT device to Edge server and Edge server to Vehicular fog \\
    \hline
    
    $\mu^{A}$, $\mu^{\text{V}}$ & Computational capacity of edge servers and vehicular fog \\
    \hline
   
    $B_{j}^{\mathrm{I \rightarrow A}}$, $B^{\mathrm{A \rightarrow V}}$ & Uplink and  Downlink transmission rate from $j$th IoT device to edge server and edge server to vehicular fog  \\
    \hline
   
    $\lambda_{j}$  & The packet arrival from the $j$th IoT device to the edge server \\
    \hline
    $S$ & Size of the original packets from the IoT device  \\
    \hline
    $Y$ & Size of the processed data packets from the edge server \\
    \hline
    $\E{\big [ R_{j} \big ]}$  & Time expected between packet arrivals from IoT devices \\
    \hline
    $\E{\big [ I_{j}^{T} \big ]}$, $\E{\big [ V_{j}^{T} \big ]}$  & Time expected for packets to be sent from IoT devices to edge server, and from edge nodes to vehicular fogs \\
    \hline
    $\E{\big [ L_{j}^{T} \big ]}$, $\E{\big [ L_{j}^{P} \big ]}$ & Data transmission queue and processing queue expected waiting times  \\
    \hline
    $\E{\big [ A_{j}^{P} \big ]}$ & The edge server's expected processing time \\
    \hline
    $Pr$  & Server busyness probability \\
    \hline
    $\Delta$  & Average end-to-end AoI  \\
    \hline
    \end{tabular}%
\end{table}%

Figure \ref{fig: Pmagt} illustrates the structure of queuing systems for packet management policies. Packets are generated from the source device, pushed into an infinite buffer in the transmission system, and then moved to the processing system. Both the transmission and processing systems are formulated as queuing models with infinite buffer sizes. These buffers are termed the transmission queue and the processing queue.
An arriving packet will be transmitted to the edge server if the transmission channel is available. However, when the transmission channel is busy, arriving packets wait in the transmission queue.
We assume that the server can handle at most one packet at a time, with packets being serviced using a first-come-first-served (FCFS) queuing discipline. When the server is busy, any new arriving packets are placed in the processing queue until the current packet is finished. 
    
\begin{figure}[!t]
  \centering
  \includegraphics[width=9cm,height=7cm]{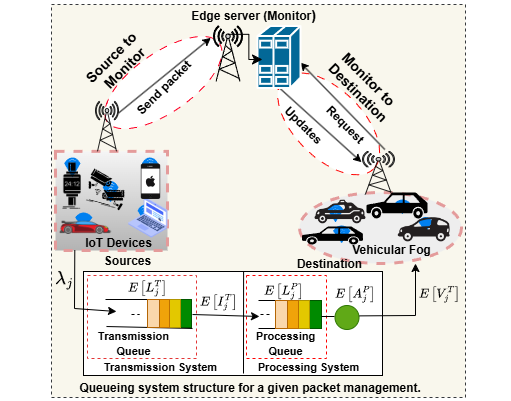}
  \caption{Illustration of the edge and vehicular fog architecture with the queueing system for a given packet management policy.}\label{fig: Pmagt}
  \end{figure}
  
\subsection{AoI Analysis in Edge-enabled Vehicular Fog Systems}
\subsubsection{Communication model}
The size of the processed packet affects the communication between the transmitted packet from the IoT source to the edge server, and the edge server to vehicular fog. It directly impacts the transmission rate and how long it takes to transmit it in the transmission stage. This communication model can be formulated using Shannon's formula \cite{shannon1948mathematical} as:
{\small
\begin{flalign}\label{eq:CM1}
   & B_{j}^{I \rightarrow A}=W\log_{2}\Bigg(1+\frac{P_{j}C(d_{j}^{IA})^{-l}}{N^{2}}\Bigg),
    \end{flalign}
\begin{flalign}\label{eq:CM2}
    & B_{j}^{A \rightarrow V}=W\log_{2}\Bigg(1+\frac{P_{A}C(d_{j}^{AV})^{-l}}{N^{2}}\Bigg),
    \end{flalign}}
where $B_{j}^{\mathrm{I \rightarrow A}}$ and $B_{j}^{A \rightarrow V}$ are the up-link and down-link rate of a packet for the communication between the IoT device to edge server, and edge server to vehicular fog, respectively, $W$ is the channel bandwidth, $C$ is the gain of a channel, $P_{j}$ and $P_{A}$ are the transmit channel power gains of the IoT device and edge server, respectively, $N^{2}$ is the density of noisy power and $-l$ is the path loss exponent.

\subsubsection{Time-Average Age Analysis}
-
\begin{figure}[!t]
  \centering
  \includegraphics[width=8.5cm,height=2.5cm]{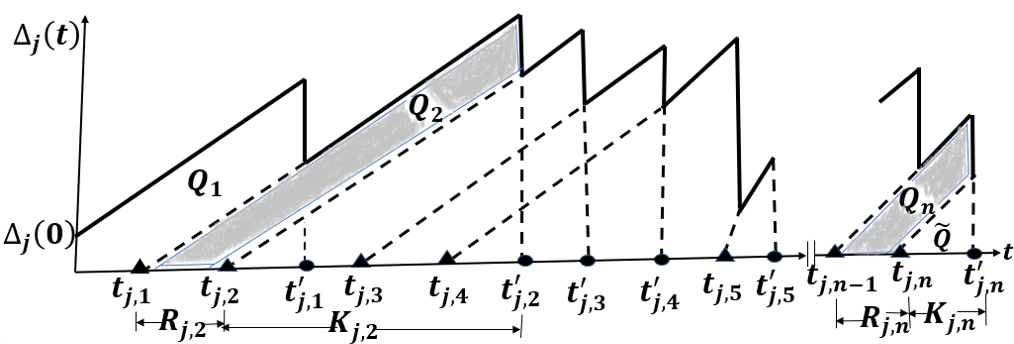}
  \caption{An example of how AoI evolves for $j$th IoT devices. Triangles and circles represent packet arrival times at edge nodes and vehicular fog, respectively.} 
  \label{fig: ExAoI}
\end{figure}

Figure \ref{fig: ExAoI} shows a sample variation in age $\Delta_{j}(t)$ as a function of time $t$ for the $j$th source at the destination. Without loss of generality, we start observing $t = 0$ when the queue is empty and the age is $\Delta_{j}(0)=0$. The initial status update of the $j$th source is time-stamped as $t_{(j,1)}$ and is followed by updates time-stamped $t_{(j,2)}$, $t_{(j,3)}$, $\cdots$, $t_{(j,n)}$. The AoI of the packet generated by the IoT device $I_{j}$ is given by $\Delta_{j}(t)=t-u_{j}(t)$, where $u_{j}(t)$ is the generation time of the most recently received packet from an IoT device $I_{j}$ until time instant $t$.

At the destination, the age of source $I_{j}$ increases linearly without updates, and it resets to a smaller value when an update is received. The update of the $j$th source, generated at time $t_{(j,i)}$, finishes service and is received by the destination at time $t_{(j,i)}^{\prime}$. At $t_{(j,i)}^{\prime}$, the age $\Delta_{j}(t_{(j,i)}^{\prime} )$ at the destination is reset to the age $K_{(j,i)}=t_{(j,i)}^{\prime}-t_{(j,i)}$ of an update on the status. $K_{(j,i)}$ denotes the system time of packet $i$ for the $j$th IoT device. Thus, the age function $\Delta_{j}(t)$ shows a saw-tooth pattern as seen in Figure \ref{fig: ExAoI}. In the status update process, the time-average age is the area under the age graph normalized by the observation period $\left(0, T\right)$, as shown in \cite{Roy2019}.
{\small
\begin{flalign}\label{eq:CM3}
    &\Delta_{j}(T)=\frac{1}{T} \int_{0}^{T} \Delta_{j}(t) \,dt.
    \end{flalign}}
    
To simplify it, the observation interval length has been set to $T=t_{(j,n)}^{\prime}$. In Equation (\ref{eq:CM3}), the area defined by the integral is divided into the sum of the areas of the polygons $Q_{1}$, the trapezoidal areas $Q_{i}$ for $i \geq 2$ ($Q_{2}$ and $Q_{n}$, highlighted in Figure \ref{fig: ExAoI}), and the triangular area of width $K_{(j,n)}$ over the time interval $\tilde{Q}$. $Q_{i}$ is the difference in the area of the points connected by the base of an isosceles triangle can be calculated $t_{(j,i-1)}$  and $t_{(j,i)}^{\prime}$, and the area of the isosceles triangle with a base connecting the points $t_{(j,i)}$ and $t_{(j,i)}^{\prime}$, then the inter-arrival time of the update can be calculated as $R_{(j,i)}=t_{(j,i)}-t_{(j,i-1)}$.
\begin{flalign}\label{eq:CM4}
    & Q_{i}=\frac{1}{2}(R_{j,i}+K_{j,i})^{2}-\frac{1}{2}K_{j,i}^{2}=\frac{1}{2}K_{j,i}^{2}+R_{j,i}K_{j,i}.
    \end{flalign}
With $N_{j}(T)=\max \left\{n \mid t_{n} \leq T\right\}$, where $N_{j}(T)$ is the number of packets of $j$th sources updated by time $T$. After some rearrangement and using Equation (\ref{eq:CM4}), the time-average age is
{\small
\begin{flalign}\label{eq:CM5}
    &\Delta_j(T)=\frac{\Bar{Q}}{T}+\frac{\left(N_j(t)-1\right)}{T} \times \frac{1}{\left(N_j(t)-1\right)} \sum_{i=2}^{N_j(T)} Q_i, 
    \end{flalign}}
where $\Bar{Q}=Q_1+\tilde{Q}$, and as $T \rightarrow \infty$ the contribution of $\Bar{Q}$ to the average AoI is negligible.

\begin{dfn}
If ($R_{(j,n)}, K_{(j,n)} $) is a stationary sequence with a marginal distribution equal to ($R_j, K_j $), then for source $j$ there exists an ergodic and stationary status updating system, and as $T\rightarrow\infty$,
  $\frac{N_j(T)}{T} \rightarrow \frac{1}{\mathbb{E}\left[R_j\right]}$, and $\frac{\sum_{i=2}^{N_j(T)} Q_i}{N_j(T)-1}\rightarrow \mathbb{E}[Q]$ with probability 1.
\end{dfn}
In a stationary ergodic status updating system, $R_j$ represents the inter-arrival time between updates from source $j$, and $K_j$ is the system time for these updates.
Little’s Law does not explicitly provide conditions for the ergodicity of the age process, but these can be verified in well-designed service systems. In systems with infinite buffers, update rates must be controlled to maintain a stationary and ergodic age process \cite{Roy2019}. For such a system, the time average $\Delta_j(T)$ tends to the probability density, and the average AoI of $j$th source using Equation (\ref{eq:CM5}) as $T \rightarrow \infty$, is
\begin{flalign}\label{eq:CM6}
   &\Delta_j=\lim _{T \rightarrow \infty} \Delta_j(T).
    \end{flalign}
The system time $K_j$ is determined by the waiting time in the processing system $K_{j}^P$ and in the transmission system $K_{j}^T$. As such, we can write $K_{j}=K_{j}^P+K_{j}^T$.
 So, the AoI for the $j$th source in Equation (\ref{eq:CM6}) is formulated as
 {\small
\begin{flalign}\label{eq:CM7}
   &\begin{aligned} & \Delta_j=\frac{\mathbb{E}[Q]}{\mathbb{E}\left[R_j\right]}=\frac{\mathbb{E}\left[R_j K_j\right]+\mathbb{E}\left[R_j^2\right] / 2}{\mathbb{E}\left[R_j\right]} \\ & =\frac{\mathbb{E}\left[R_j K_j^P\right]+\mathbb{E}\left[R_j K_j^T\right]+\mathbb{E}\left[R_j^2\right] / 2}{\mathbb{E}\left[R_j\right]}.\end{aligned}
    \end{flalign}}
Referring to Equation (\ref{eq:CM7}), we emphasize to $R_j$ and $K_j$, it should  noted that the packet arrives from the $j$th IoT device with parameter $\lambda_{j}$ implies $\mathbb{E}\left[R_j\right]=\frac{1}{\lambda_j}$ and $\mathbb{E}\left[R_j^2\right]=\frac{2}{\lambda_j^2}$, then the AoI can be formulated as
\begin{flalign}\label{eq:CM8}
   &\Delta_j=\frac{1}{\lambda_j}+\mathbb{E}\left[K_j^P\right]+\mathbb{E}\left[K_j^T\right].
    \end{flalign}
Following Equation (\ref{eq:CM8}), 
the waiting time for a packet in the processing system is determined by its waiting time in the processing queue ($L_j^P$) plus the server processing time ($A_j^P$).
This can be expressed as:  $\mathbb{E}\left[K_j^P\right]=\mathbb{E}\left[L_j^P\right]+\mathbb{E}\left[A_j^P\right]$, where $\mathbb{E}\left[A_j^P\right] =  \frac{1}{\mu^{A}}$. 
In the other case, the waiting time for a packet in the transmission system is given by the sum of its waiting time in the transmission queue ($L_j^T$) and the time taken in the transmission channel, which includes the transmission from the queue to the edge server ($I_j^T$) and from the edge server to the vehicular fog ($V_j^T$).
This can be expressed as: $\mathbb{E}\left[K_j^T\right]=\mathbb{E}\left[L_j^T\right]+\mathbb{E}\left[I_j^T\right]+\mathbb{E}\left[V_j^T\right]$. Then Equation (\ref{eq:CM8}) can be reformulated as
{\small
  \begin{equation}
\begin{aligned}\label{eq:CM9}
& \Delta_j=\frac{1}{\lambda_j}+\frac{1}{\mu^A}+\mathbb{E}\left[L_j^P\right]+\mathbb{E}\left[I_j^T\right]+\mathbb{E}\left[V_j^T\right]+\mathbb{E}\left[L_j^T\right].
\end{aligned}
\end{equation}}
Following this, we need the analyze of each element, in Equation (\ref{eq:CM9}), $\left(\text{i.e., } \mathbb{E}\left[L_j^P\right],\mathbb{E}\left[I_j^T\right],\mathbb{E}\left[V_j^T\right], \mathbb{E}\left[L_j^T\right]\right)$.

\subsubsection{ Analysis of Packet transmission rate}
To analyze the packet transmission rate, as formulated in Equations (\ref{eq:CM1}) and (\ref{eq:CM2}), we utilized Shannon's formula \cite{shannon1948mathematical}. Here the rate stands for the transmission channel capacity. Higher transmission rates allow faster data transmission, reducing the time it takes for information to travel from the source to the destination. Faster transmission rates lead to lower AoI, as updates can be delivered more quickly and frequently, keeping the information fresher.

\paragraph {Analysis of $\mathbb{E}\left[I_j^T\right]$}
The transmission rate is given in Equation (\ref{eq:CM1}), and according to a Rayleigh distribution \cite{weinberg2016optimal}, transmission rates are random variables. Then we can calculate the expected service time for packets from the $j$th IoT device using the following theorem.
\begin{theorem}\label{Th_1} The time it takes to successfully transmit packets generated by $j$th IoT device is calculated as
\begin{equation} 
\begin{aligned}\label{eq:CM10}
\resizebox{0.4\textwidth}{!}{$\mathbb{E}\left[I_j^T\right]=\frac{\operatorname{S_j\ln(2)} N^2 d_j^{I A} l}{W P_j} \int_0^{\infty} \exp \left(\frac{\operatorname{S_j\ln(2)}}{W t}+\frac{1-\exp \left(\frac{\operatorname{S_j\ln(2)}}{W t}\right)}{P_j N^{-2} d_j^{I A}-l}\right) \frac{d t}{t}$ }.
\end{aligned}
\end{equation}
\end{theorem}

\begin{proof} The proof is shown in Appendix A
\end{proof}

\paragraph {Analysis of $\mathbb{E}\left[V_j^T\right]$ }
Given the transmission rate of processed packets from the edge server to vehicular fogs in Equation (\ref{eq:CM2}) and considering a Rayleigh distribution \cite{weinberg2016optimal} that the transmission rate is a random variable. Based on the following theorem, we can determine how long the data packet should take from the edge server to the vehicular fog during the transmission subsystem.

\begin{theorem}\label{Th_2}
The time it takes to successfully transmit processed packet from the edge server to vehicular fog is calculated as
\begin{equation}
\begin{aligned}\label{eq:CM11}
\resizebox{0.4\textwidth}{!}{$\mathbb{E}\left[V_j^T\right]=\frac{Y_j\ln(2)N^2 d^{A V^l}}{W P_A}\int_0^{\infty} \exp \left(\frac{Y_j\ln(2)}{W t}+\frac{1-\exp \left(\frac{Y_j\ln(2)}{W t}\right)}{P_A N^{-2} d^{A V^{-l}}}\right) \frac{d t}{t}$}.
\end{aligned}
\end{equation}
\end{theorem} 
\begin{proof} The proof is shown in Appendix B
\end{proof}

\subsubsection{Analysis of $\mathbb{E}\left[L_j^P\right]$ }
A freshly received packet must wait in the processing queue until the completion of the packets currently on the server. Let $Pr_B$ be the probability of busy periods on the server. Using Little’s law \cite{harchol2013performance}, we can compute $Pr_B$ as
\begin{equation}
\begin{aligned}\label{eq:CM12}
\resizebox{0.3\textwidth}{!}{$P r_{\mathbf{B}}=\sum_{j=1}^J \lambda_j \mathbb{E}\left[A_j^P\right]=\sum_{j=1}^J \lambda_j \frac{1}{\mu^A}$}.
\end{aligned}
\end{equation}
Based on this analysis, we can formulate the expected waiting time $\left(\mathbb{E}\left[L_j^P\right]\right)$ of a packet in the processing queue using the following theorem.

\begin{theorem}\label{Th_3}
In the processing queue for packets from $j$th IoT device, the waiting time is formulated as
\begin{equation}
\begin{aligned}\label{eq:CM13}
\mathbb{E}\left[L_j^P\right]=\frac{\Sigma_{j=1}^J \rho_j^2 / \lambda_j}{2\left(1-\Sigma_{i=1, i\neq j}^J \rho_i\right)\left(1-\beta \Sigma_{i=1, i\neq j}^{J-1} \rho_i\right)},
\end{aligned}
\end{equation}
\end{theorem}
where $\beta$ is indicator variable,$\beta= \begin{cases}1, & \text { if } J>1, \\ 0, & \text { if } J=1,\end{cases}$\\
$\rho_j=\lambda_j \mathbb{E}\left[A_j^P\right]$
the load caused by data packets from $j$th IoT device.
\begin{proof} The proof is shown in Appendix C
\end{proof}

\subsubsection{Analysis of $\mathbb{E}\left[L_j^T\right]$ }
We take into consideration M/G/1 systems \cite{shore1982information}, in which customers come from an infinite customer pool with independent, exponentially distributed inter-arrival times, wait in a queue with an infinite capacity, are independently served by a single server with a general service time distribution, and then go back to the client pool \cite{prabhu1987bibliography}.
The estimated waiting time spent in the transmission queue was estimated by using the principle of maximum entropy (PME) \cite{shore1982information}. Based on this principle, we can derive the following theorem.
\begin{theorem}\label{Th_4}
In the transmission queue, the estimated time spent waiting for packets originally generated from the $j$th IoT device is calculated as
{\small
\begin{equation}
\begin{aligned}\label{eq:CM14}
\mathbb{E}\left[L_j^T\right]=\frac{\left(\Sigma_{j=1}^J \lambda_j \mathbb{E}\left[I_j^T\right]\right)^2}{\sum_{i=1, i\neq j}^J \lambda_i\left(1-\Sigma_{i=1, i\neq j}^J \lambda_i \mathbb{E}\left[I_j^T\right]\right)}.
\end{aligned}
\end{equation}}
\end{theorem}

\begin{proof} The proof is shown in Appendix D
\end{proof}

\section{Problem Formulation and Proposed Solution}
\label{sec:Problem_Formulation}

\subsection{AoI Optimization Problem Formulation}
We considered an edge-enabled vehicular fog system, which consists of multiple IoT devices, an edge server, and a vehicular fog. We aim to minimize the average end-to-end AoI ($\Delta)$ of edge-enabled vehicular fog systems. To achieve this, we formulated the average end-to-end AoI ($\Delta)$ as
\begin{equation}
\Delta=\frac{1}{T \times J} \sum_{j=1}^J \Delta_j,
\end{equation}
where $\Delta_j$ indicates the average end-to-end AoI of $j$th IoT device updates over the long term, given in Equation (\ref{eq:CM9}), $\lambda_J$ is the updated arrival rates, and $T$ the total time slot. Hence, the corresponding problem is stated as
\begin{equation}
\min _{\lambda_J}(\Delta),
\end{equation}
subject to
\begin{subequations}
\begin{equation}
\begin{aligned}\label{eq:CM16a}
\quad \sum_{j=1}^J \lambda_j \mathbb{E}\left[A_j^P\right]<1,
\end{aligned}
\end{equation}
\begin{equation}
\begin{aligned}\label{eq:CM16b}
\sum_{j=1}^J \lambda_j \mathbb{E}\left[I_j^T\right]<1,
\end{aligned}
\end{equation}
\begin{equation}
\begin{aligned}\label{eq:CM16c}
\sum_{j=1}^J \lambda_j \mathbb{E}\left[V_j^T\right]<1,
\end{aligned}
\end{equation}
\begin{equation}
\begin{aligned}\label{eq:CM16d}
\lambda_j>0 \quad \forall j \in J.
\end{aligned}
\end{equation}
\end{subequations}
The constraint in Equation (\ref{eq:CM16a}) guarantees the stability of the processing capacity, which is only one packet at a time. The inequalities in Equations (\ref{eq:CM16b}) and (\ref{eq:CM16c}) also ensure that the transmission channel will serve only one packet at a time (transmission from IoT device to edge server and, from the edge server to vehicular fog), respectively. The constraint in Equation (\ref{eq:CM16d}) guarantees the non-negativity of each packet arrival.

\subsection{Proposed Solution}
Our objective is to obtain the optimal policy that specifies the actions taken at different states of the system over time, achieving the minimum average end-to-end AoI at the destination. To solve this optimization problem, we used a modified DRL function. This function is known as an agent, which interacts with the edge-enabled vehicular fog environment to find optimal solutions. The agent can be placed on an edge server to offload information to vehicular fogs. A DRL agent is trained using a replay memory buffer that stores information in the form of, states, actions, rewards, and next-states $(s_t,a_t,r_t,s_t^\prime)$. The agent learns the optimal action to take given a particular environmental state. The state represents the current condition of the edge-enabled vehicular fog system. The probability of offloading processed packets to vehicular fog is the action taken by the DRL agent. The optimal action aims to minimize the average end-to-end AoI.
\subsubsection{Edge-enabled Vehicular Fog Environment setup} 
The offloading decisions in edge-enabled vehicular fog environments require appropriate decisions to be made at each moment in time to achieve their objectives. We have implemented DRL-based agents to interact with the environment and make quick decisions. We define the agent, states, actions, and rewards for AoI optimization in the designed environment as follows.
\paragraph{\textbf{Agent}} Agent is a decision-maker who interacts with edge-enabled vehicular fog environments.

\paragraph{\textbf{Environment}} Edge-enabled vehicular fog systems topology ($I_j,A,V$).

\paragraph{\textbf{State ($s_{t}$)}} State is a representation of the current environment that the agent can observe. It includes all relevant information necessary for decision-making, such as current topology ($I_j, A, V$), packet arrival ($\lambda_j$), computation resources for edge server ($\mu_j^A$), and vehicular fog ($\mu^V$), communication resources between IoT devices and edge server ($B_{j}^{\mathrm{I \rightarrow A}}$) and edge server and vehicular fogs ($B^{\mathrm{A \rightarrow V}}$), and the AoI ($\Delta$). The state arrays are constructed from the states and used as inputs to the algorithms: {\small $AP_t=\left\{\lambda_j\right\}_{j=1}^J$, $C_t=\left\{\left\{\mu_j^A\right\}_{j=1}^J, \mu^{V}\right\}$, $B_t=\left\{\left\{B_j^{I \rightarrow A}\right\}_{j=1}^J , B^{A \rightarrow V}\right\}$, $\Delta_t=\{\Delta\}_{j=1}^J$},
where $AP_t$ is the set of packet arrival rates to all edge servers at time $t$. $C_t$ is an array of the computing capacities at the edge server and vehicular fogs, and $B_t$ is the set of communication capacities between IoT devices, edge servers, and vehicular fogs. $\Delta_t$ set of the AoI of the system. $AP_t$, $C_t$, $B_t$, and $\Delta_t$ are combined to form the environment state at time $t$. Then the system environment state $s_t$ at time $t$ can be represented as $s_t=\left[A P_t, C_t, B_t, \Delta_t\right]$.
\paragraph{\textbf{Action}} An action refers to the potential offloading decisions that the agent makes and applies in the edge-enabled vehicular fog environment. Since the offloading decision is binary, then  $P_j^{A \rightarrow V}\in\left\{0,1\right\}$. At time $t$ the action $a_{t}$  can be expressed as $a_{t}=\left\{P_j^{A \rightarrow V}\right\}_{j=1}^{V}$. Let $\mathcal{A}$ be a set of offloading decisions to be sent to the requested vehicular fogs, where $\mathcal{A} =\left\{a_{1,t}, \cdots,a_{V,t}\right\}$.

\paragraph{\textbf{Reward ($r_{t}$)}} Our objective is to maximize the cumulative rewards $r_t$, it receives over time while minimizing end-to-end AoI. So, the reward function $\lim _{T \rightarrow \infty} \frac{1}{T} \sum_{t=0}^T r_t$ should be related to the goal of our optimization problem, which is given as
{\small
\begin{equation}
\begin{aligned}\label{eq:CM17}
r_t=\left(\frac{\Delta_t}{J \times T}\right),
\end{aligned}
\end{equation}}
where $t$ is the episode at $t$th time.
\subsubsection{Proposed Dueling-DQN Algorithm}
The dueling-DQN algorithm is a model-free DRL approach that predicts future states and rewards within an environment without learning the transition function of the environment \cite{wang2016dueling}. The primary objective of this algorithm is to determine the best policy. To achieve this, the neural network used in the dueling-DQN algorithm divides its last layer into two sections: one section estimates the expected reward collected from a specific state, called the state value function ($V(s_{t})$), while the other section measures the relative benefit of one action compared to others, called the advantage function ($A(s_{t}, a_{t})$). In the end, it merges the two components into a single output, which estimates a state-action value $Q(s_{t},a_{t})$. It can be computed as

\begin{equation}
\begin{aligned}\label{eq:CM18}
Q^\pi(s_{t}, a_{t})=V^\pi(s)+A^\pi(s_{t}, a_{t}).
\end{aligned}
\end{equation}
The expectation of the advantage function is zero, and the value function $Q$ in Equation (\ref{eq:CM18}) is replaced using Equation (\ref{eq:CM181}), where $A(s_{t}, a_{t}^\prime)$ is a $|\mathcal{A}|$-dimensional vector for a set of actions
{\small
\begin{equation}
\begin{aligned}\label{eq:CM181}
Q^\pi(s_{t}, a_{t})=V^\pi(s_{t})+\left(A^\pi(s_{t}, a_{t})-\frac{1}{|\mathcal{A}|} \sum_{a_{t}^{\prime}} A\left(s_{t}^\prime, a_{t}^{\prime}\right)\right),
\end{aligned}
\end{equation}}
where policy $\pi$ is the distribution connecting state to action. Let $Q(s_{t},a_{t};\theta,\phi,\varphi)$ be the value function with parameters $\theta$, which is expressed as
\begin{equation}
\begin{aligned}\label{eq:CM19}
Q(s_{t}, a_{t} ; \theta, \emptyset, \varphi)=V(s_{t} ; \theta, \emptyset)+A(s_{t}, a_{t} ; \theta, \varphi),
\end{aligned}
\end{equation}
where $\theta$ is parameter of the convolutions layer, $\emptyset$ is the action-value function parameter of fully-connected, $\varphi$ is advantage parameter of fully-connected.
The dueling DQN agent observes and interacts with the environment. The target $Q$-value ($Q_t$) of the target $Q$-network given as
\begin{equation}
\begin{aligned}\label{eq:CM20}
Q_t=r_{t}+\gamma \max _{a_{t}^{\prime}}{Q}\left(s_{t}^{\prime}, a_{t}^{\prime}; \theta_{\text {targ}}, \emptyset, \varphi\right).
\end{aligned}
\end{equation}
Then, the loss function $L(\theta)$ in dueling-DQN is
\begin{equation}
\begin{aligned}\label{eq:CM21}
L(\theta)=\mathbb{E}\left[\left(Q_t-Q(s_{t}, a_{t} ; \theta, \emptyset, \varphi)\right)^2\right],
\end{aligned}
\end{equation}
where $\mathbb{E}$ denotes the expected value. Finally, the action-value function $Q(s_{t},a_{t};\theta,\emptyset,\varphi)$ is updated as
{\small
\begin{equation}
\begin{aligned}\label{eq:CM22}
Q(s, a ; \theta, \emptyset, \varphi)=Q(s_{t}, a_{t} ; \theta, \emptyset, \varphi)+\mathbb{E}\left[Q_t-Q(s_{t}, a_{t} ; \theta, \emptyset, \varphi)\right].
\end{aligned}
\end{equation}}

\begin{algorithm}[!t]
\caption{Proposed Dueling-DQN Algorithm}
\label{alg:DuelingDQN}
   \LinesNumbered
    Input: A set of IoT devices, computing capacity, vehicular fog, and packet arrival.

    Output: The target $Q$-value ($Q_{t}$) at each time step.
     
    Initial current $Q$-function parameters: $Q(s_{t}, a_{t}; \theta)$
 
    Initial target $Q$-function parameters: $Q\left(s_{t}^{\prime}, a_{t}^{\prime}; \theta_{\text {targ}}\right)$
      
    Initialize replay memory $\mathcal{M}$ to capacity $N$,  
    $Episode\_max$, $Policy\_update\_frequency$,  
    $Training\_step =T$, $Total\_step$, set $\max$ and initial $\varepsilon$

    Set target parameters equal to main parameters $\theta_{\text {targ }} \leftarrow \theta$
          
    \For{$Initial\_episode =1,$ to $Max\_Episode$}{
        done $=$ False
              
        Initialize the starting state $s_{0}$
        
       \For{$t=1,$ to $T$}{
            Observe the environment state $s_{t}$,
            with probability $\varepsilon$ select a random action $a_t$
            otherwise select $a_t=\underset{a_{t}}{\operatorname{argmax}} Q\left(s_t, a_{t} ; \theta\right)$
                  
            Execute $a_{t}$ in the environment

            Calculate reward $r_{t}$ using Equation (\ref{eq:CM17})
            
            Observe next state $s_{t}^{\prime}$, and done signal $d$ to indicate whether $s_{t}^{\prime}$ is terminate;
                  
            Store $\left(s_{t}, a_{t}, r_{t}, s_{t}^{\prime}, d\right)$ in replay memory $\mathcal{D}$
                  
            $s_{t} \leftarrow s_{t}^{\prime}$
                  
           \If {it's time to update}{
                  
               \For {i in range(Max\_Episode)}{  \# mini-batch process
                    randomly sample a batch of transitions, initial $\varepsilon$, $B=\left(s_{i}, a_{i}, r_{i}, s_{i}^{\prime}, d\right)$ from $\mathcal{M}$

                    Compute targets\\ 
                    $Q_{t}=\begin{cases}\text{Equation}(\ref{eq:CM17}),&\text{if episode terminates }\\ 
                    \text{Equation}(\ref{eq:CM20}),&\text{otherwise.}
                    \end{cases}$
                        
                    Reset gradient $L\left(\theta\right)$ to 0
                        
                    Perform a gradient descent step on according to Equation (\ref{eq:CM21}), and update  
                    $\theta \leftarrow \theta+\alpha\nabla_\theta L\left(\theta\right)$
                        
                    \If {$Total\_step \mod $ $Policy\_update\_frequency$ ==$0$}{
                        Update action-value function networks using Equation (\ref{eq:CM22})
                    
                    }
                          
                    $\theta_{\text {targ}} \leftarrow \theta$
                          
                    Update $\varepsilon$
                }
            }
        }
   }
\end{algorithm}

\begin{figure}[!t]
  \centering
  \includegraphics[width=9cm,height=7cm]{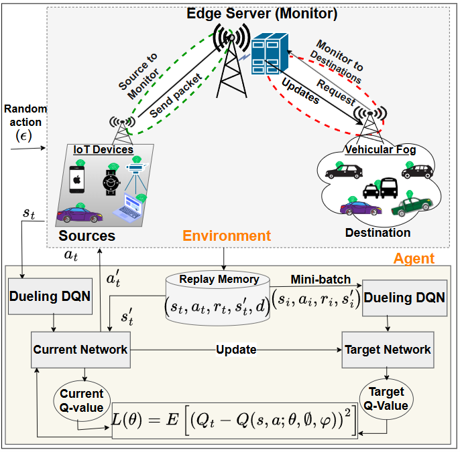}
  \caption{The exchange of information between the environment and the proposed dueling-DQN algorithm.}\label{fig:algo}
\end{figure}

To achieve optimal offloading of real-time updating information in an edge-enabled vehicular fog environment, we propose Algorithm \ref{alg:DuelingDQN}, utilizing a modified dueling-DQN algorithm. The dynamic changes in the environment bring about variations in agent actions and state spaces at each step of the process. Due to improved network stability and faster convergence time, managing the dynamics of this process is possible through the proposed algorithm. After several iterations, an end-to-end AoI optimization policy can be obtained to deliver the most optimally updated information to vehicular fog. 
Algorithm \ref{alg:DuelingDQN} works as follows:

 \paragraph {Step 1}The algorithm takes a set of IoT devices, computing capacity, vehicular fog, and packet arrival as input state to return an offloading decision as an output, as shown in lines 1 and 2. The current and target $Q$-function parameters, replay memory, the maximum number of episodes, policy update frequency, the maximum time slots, and the maximum Epsilon greedy are initialized in lines 3-6.
\paragraph {Step 2} In the agent's training process, at each time slot the agent observes the environment state $s_{t}$ and chooses the random action $a_{t}$ with a probability of epsilon-greedy. Then execute the action in the environment to calculate the reward $r_{t}$ and observe the next state $s_{t}^{\prime}$ as shown in lines 11-14. The transition $\left(s_{t}, a_{t}, r_{t}, s_{t}^{\prime}, d\right)$ is stored in the replay memory, and the current state is set to the next state in lines 15-16.
\paragraph {Step 3} To train the policy network, a batch of transition samples is randomly drawn from the mini-batch ($B$), and then, the target network is computed according to Equations (\ref{eq:CM17}) and (\ref{eq:CM20}) in lines 19-21.
\paragraph {Step 4} The gradient of the network parameters, denoted as $L\left(\theta\right)$, is computed using Equation (\ref{eq:CM21}). The network parameter $\theta$ is then updated using the gradient $L\left(\theta\right)$, following the expression $\theta \leftarrow \theta + \alpha\nabla_\theta L\left(\theta\right)$ and then update the action-value network using Equations (\ref{eq:CM22}) in lines 23-28.

Figure \ref{fig:algo} illustrates the exchange of information between the environment and the proposed dueling-DQN algorithm. It comprises three primary components: First, the edge-enabled vehicular fog environment provides information to the agent as a state ($s_{t}$) to determine the offloading decisions ($a_{t}$) at time $t$ iterations. Secondly, the replay memory serves as a storage for transition information, capturing the history of agent-environment interactions as a tuple of state, action, reward, and next state $\left(s_{t}, a_{t}, r_{t}, s_{t}^{\prime}, d\right)$. This transition information is crucial for updating the agent model’s parameters during training. Finally, the agent, the offloading agent utilizes the current network to update both the current $Q$-value and the target network $Q$-value. This updating process is done based on the gradient loss values.

\section{Performance Analysis}
\label{Performance}
\subsection{Simulation Parameter Settings}
In edge-enabled vehicular fog systems, the Python programming language was utilized to create the simulation environment. We developed the edge-enabled vehicular fog environment, following OpenAI Gym's standards \cite{OpenJm16}. To ensure its effectiveness in meeting all required DRL criteria, we conducted rigorous testing using the Stable-Baselines3 DRL implementation \cite{BaseLine321}. This evaluation allowed us to thoroughly assess the system's performance and confirm its suitability for comparison with other DRL algorithms. The proposed solutions were compared with the existing DQN algorithm and analytical results. The parameter configurations were presented as both system parameters and agent parameters.

\begin{table}[!t]
  \centering
  \caption{System parameter settings}
  \label{tab: sysParm}
    \begin{tabular}{|p{4.5cm}|p{3.25cm}|}
    \hline
    \textbf{Parameters descriptions } & \textbf{Value(s)}  \\
    \hline
    Number of Edge server ($A$) & \multicolumn{1}{l|}{1} \\
    \hline
    Number of Vehicular fog ($V$)    & \multicolumn{1}{l|}{1} \\
    \hline
    Number of IoT Devices ($I_j$) & [$2, 3, \cdots, 10$] \\
    \hline
    Distance between $j$th IoT device and Edge server ($d_{j}^{IA}$) & 3KM \\
    \hline
    Distance between edge node and Vehicular-fog ($d^{AV}$) & 10km \\
    \hline
    System Bandwidth ($W$) & 100KHz  = ($100\times10^3  bps$) \\
    \hline
    Background noise ($N$) & -174 dBm/Hz = -124dBm \\
    \hline
    Transmission and Reception power ($p_j, P_A$) & 1000 mW, = (200dBm) \\
    \hline
    Path loss exponent ($l$) & \multicolumn{1}{l|}{3} \\
    \hline
    Time slots ($T$) & [1-10] \\
    \hline
    The capacity of the Edge server ($\mu^A$) & $30Mb/s = 30\times10^4 bps$ \\
    \hline
    The capacity of Vehicular Fog ($\mu^V$) & $15Mbps = 15\times10^4 bps$ \\
    \hline
    Arrival packets from IoT devices ($\lambda_j$) & [10, 20, 40, 60, 80]Mbps \\
    \hline
    The original packet size arrived at the edge server ($S$) & $S=\lambda_j$ \\
    \hline
    Processed packet size from the edge server to vehicular fog ($Y$) & $Y= S\times0.2$ \\
    \hline
    Packet size & 10bits \\
    \hline
    \end{tabular}%
\end{table}%

\subsubsection{System parameter settings}
Table \ref{tab: sysParm} provides system parameter settings, and we examined the effects of IoT device size and different time slots. In the experiments, we considered five arrival traffic rate scenarios. The traffic generated by the IoT devices follows a Poisson distribution, with mean values. The distances between IoT devices and the edge server and between the edge server and vehicular fog are randomly distributed as defined in Table \ref{tab: sysParm}.

\begin{table}[!t]
  \centering
  \caption{Agent hyper-parameter settings}
  \label{tab: HParm}
    \begin{tabular}{|p{4cm}|p{1.5cm}|}
    \hline
    \textbf{Parameters} & \multicolumn{1}{p{8.225em}|}{\textbf{Value}} \\
    \hline
    Learning rate ($\alpha$) & 0.0001 \\
    \hline
    Batch Size & 128 \\
    \hline
    Replay memory size & 100000 \\
    \hline
    Discount factor ($\gamma$) & 0.99 \\
    \hline
    Initial epsilon ($\epsilon$) & 0.5 \\
    \hline
    Max-Episodes &  $10^4$\\ 
    \hline
    Hidden Layer\_1 and Layer\_2 & 256 and 256 \\
    \hline
    \end{tabular}%
\end{table}%

\subsubsection{Agent hyper-parameter settings}
The simulation environment was written in Python, and for DRL policy both (dueling-DQN and DQN) algorithms, the ANNs are implemented using TensorFlow \cite{Abadi2016}. The configurations for the DRL algorithms can be found in Table \ref{tab: HParm}.

\begin{figure}[!t]
    \subfloat[Optimal episode reward of dueling-DQN policy.]{
    \includegraphics[width=0.20\textwidth]{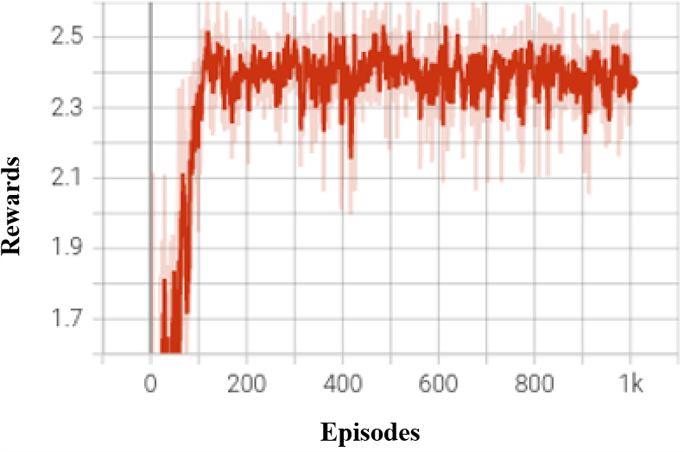}
    \label{fig: RDDQNa}}
    \qquad
    \subfloat[Optimal episode reward of DQN policy.]{
    \includegraphics[width=0.20\textwidth]{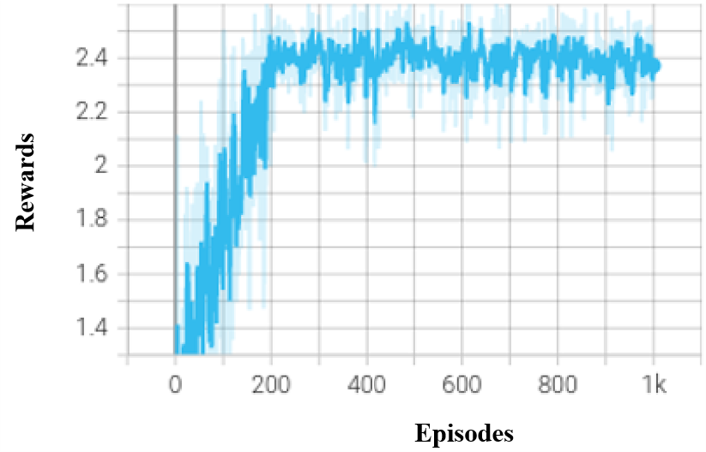}
    \label{fig: RDQNb}}
    \\
    \subfloat[The average convergence of Rewards for DRL Agents.]{
    \includegraphics[width=0.45\textwidth]{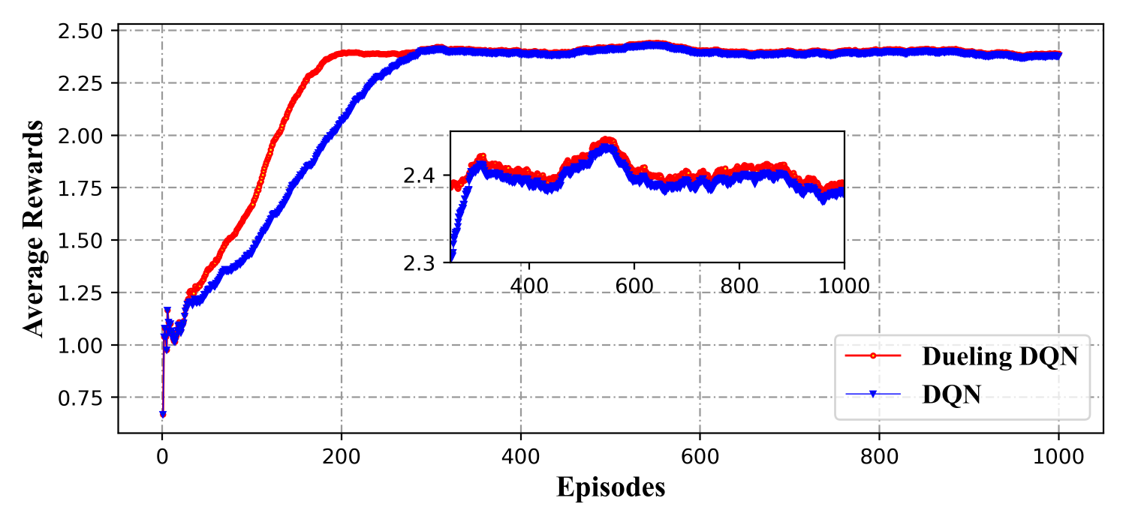} 
    \label{fig: RDRLe}}
    \qquad
    \subfloat[The average convergence AoI for DRL Agents.]{
    \includegraphics[width=0.45\textwidth]{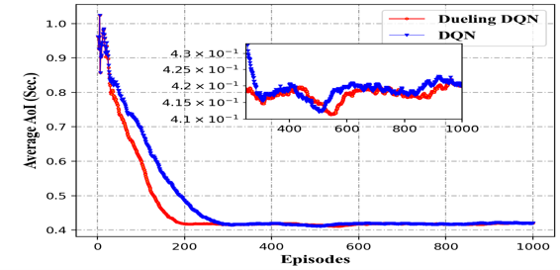}
    \label{fig: AoIDRLf}}
    \qquad
    \\
    \subfloat[Optimal loss of dueling-DQN algorithm.]{
    \includegraphics[width=0.20\textwidth]{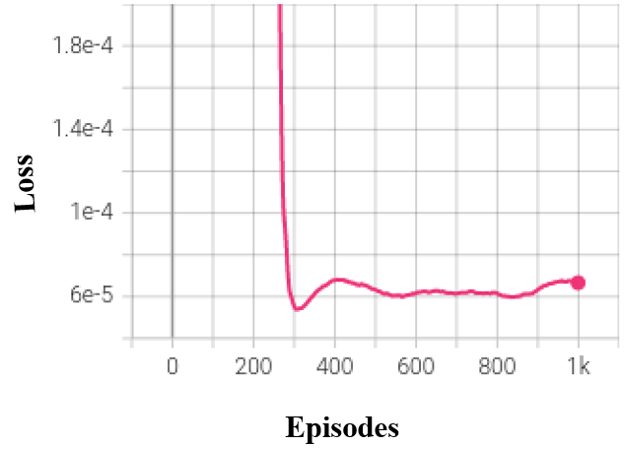} 
    \label{fig: LosDDQNc}}
    \qquad
    \subfloat[Optimal loss of DQN algorithm.]{
    \includegraphics[width=0.20\textwidth]{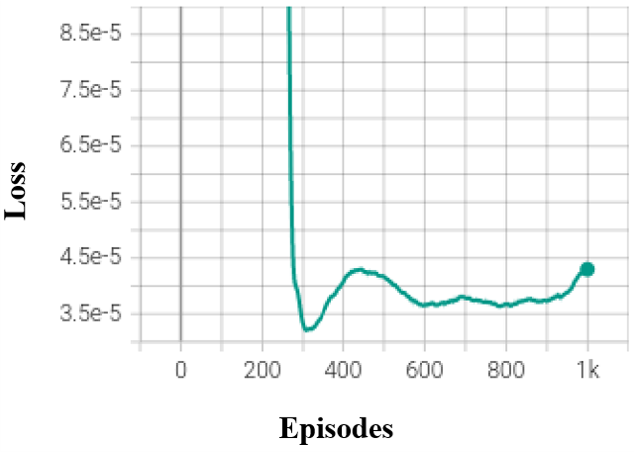} 
    \label{fig: LosDQNd}}
    \qquad
    \subfloat[Optimal average loss of DRL agents algorithm.]{
    \includegraphics[width=0.45\textwidth]{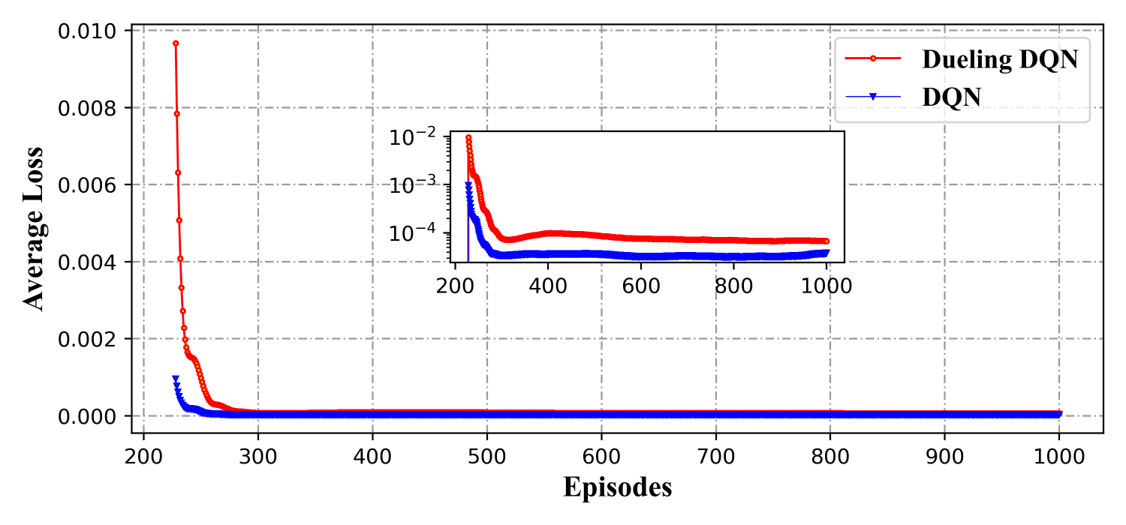} 
    \label{fig: LosDRLg}}
    \caption{The evaluations of convergence of DRL agents.}
    \label{fig: Convergence}
    \end{figure}

\subsection{DRL Agents Convergence Evaluations}

In Figure \ref{fig: Convergence}, we evaluate the performance of the proposed dueling-DQN-based DRL algorithm by assessing its performance regarding reward, loss, and the optimal AoI values achieved by the algorithm. During the evaluation of agent performance, each algorithm is run for 10,000 episodes, and the results are averaged.
The optimal episode rewards from the dueling-DQN and DQN algorithms are shown in Figures \ref{fig: RDDQNa} and \ref{fig: RDQNb}, respectively, where, in contrast to DQN, the proposed dueling-DQN algorithm learns faster and achieves optimal rewards more quickly. The average reward convergence of dueling-DQN and DQN agents is illustrated in Figure \ref{fig: RDRLe}. This shows that the dueling-DQN agent outperforms in terms of optimal rewards over a range of episodes, from 50 to 250, because the dueling-DQN agent can learn a more optimal policy.
However, after converging, both the dueling DQN and DQN agents exhibit similar performance. This similarity arises because, in both algorithms, the epsilon-greedy approach introduces randomness, leading to sub-optimal policy performance. As a result, achieving optimal reward values can be accomplished by reducing the value of epsilon to zero after the algorithm converges. In Figure \ref{fig: AoIDRLf}, the optimal convergence of AoI for DRL agents is illustrated, demonstrating that the dueling-DQN agent exhibits better convergence and optimal AoI values. Due to the utilization of randomness in both algorithms, their convergence performance becomes similar after the agents have converged. 
The losses of the dueling-DQN and DQN agents are shown in Figures \ref{fig: LosDDQNc} and \ref{fig: LosDQNd}, respectively. 
These losses are further averaged and presented in Figure \ref{fig: LosDRLg}. The dueling-DQN agent exhibits fewer outliers and improved stability because DQN is vulnerable to overestimation issues, causing worse training stability and more outliers.

Figure \ref{fig: Eval} illustrates the assessment of the dueling-DQN algorithm concerning different numbers of IoT devices and time slots. In Figure \ref{fig: AOIDIOT}, we analyze the average end-to-end AoI as the number of IoT devices increases while keeping time slots fixed. Conversely, in Figure \ref{fig: AoIDTime}, we assess the average end-to-end AoI with varied time slots while maintaining a constant number of IoT devices. Both Figures \ref{fig: AOIDIOT} and \ref{fig: AoIDTime} demonstrate that the average end-to-end AoI rises with an increasing number of IoT devices and varied time slots. The convergence of average end-to-end AoI evaluation occurs at episode number 200, and with the increase in the number of episodes, there is a noticeable trend towards greater stability and a decrease in AoI.
Based on this result, we can conclude that the dueling-DQN agent is more transferable and that the network shows more training stability. This is due to the dueling-DQN algorithm's ability to learn the dynamics of the environment and make wise decisions.

\begin{figure}[!t]
    \centering
    \subfloat[Evaluation of the dueling-DQN algorithm with an increasing number of IoT devices over the episode, ($I \in \{2, 4, 6, 8, 10\}$, $T$=10, $\lambda_j$ is used 40Mbps)]{
    \includegraphics[width=0.45\textwidth]{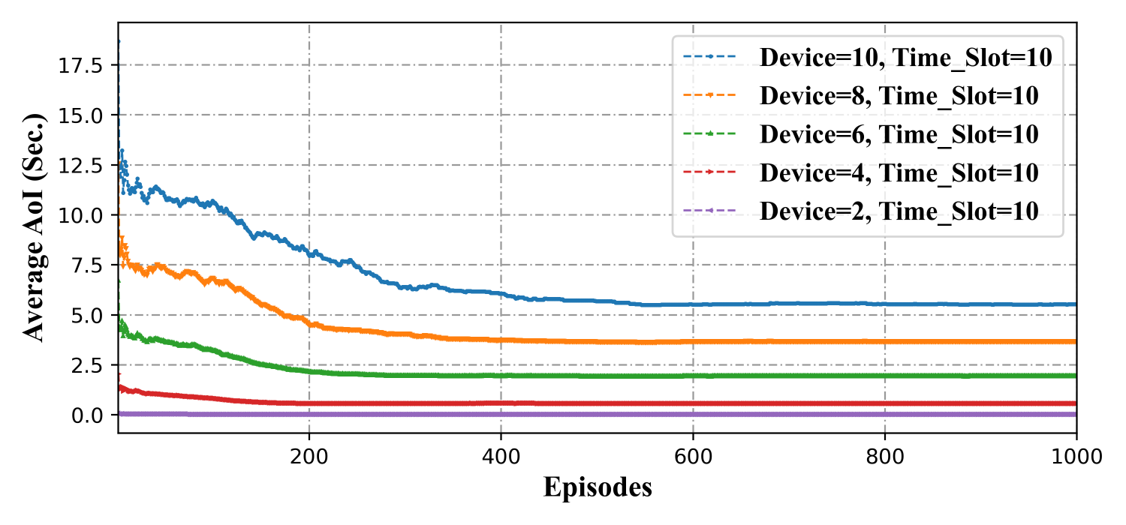}
    \label{fig: AOIDIOT}}
    \qquad
    \subfloat[Evaluation of the dueling-DQN algorithm with different time slots, ($T \in \{2, 4, 6, 8, 10\}$, $I=10$ $\lambda_j$ is used 40Mbps)]{
    \includegraphics[width=0.45\textwidth]{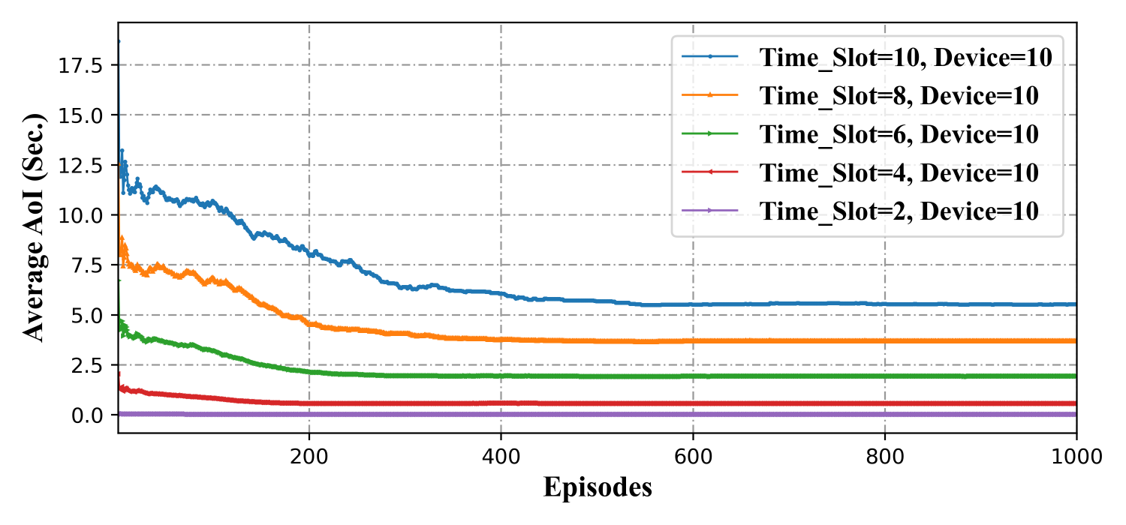}
    \label{fig: AoIDTime}}
    \caption{Evaluation of the dueling-DQN with the increasing number of IoT devices and time slots over the episodes}
    \label{fig: Eval}
\end{figure}    

\subsection{ AoI Performance Analysis}

In Figure \ref{fig: DAOI}, we compared the average end-to-end AoI performance of the system algorithms under an increasing number of IoT devices. The average end-to-end AoI has increased with the rise in the number of IoT devices for all algorithms. This is because as the number of IoT devices increases, the generation time and waiting time in the transmission channel will also increase for each packet. As a result, the average end-to-end AoI value will be higher. However, the dueling-DQN algorithm is shown to perform close to analytical results and outperform the DQN algorithm when the number of IoT devices ranges from 2 to 7. However, when the number of IoT devices exceeds 7, dueling-DQN performs even better. This improvement can be attributed to the fact that dueling-DQN learns value functions without considering actions in each state, resulting in a significant enhancement in network stability and convergence speed.

\begin{figure}[!t]
  \centering
  \includegraphics[width=0.45\textwidth]{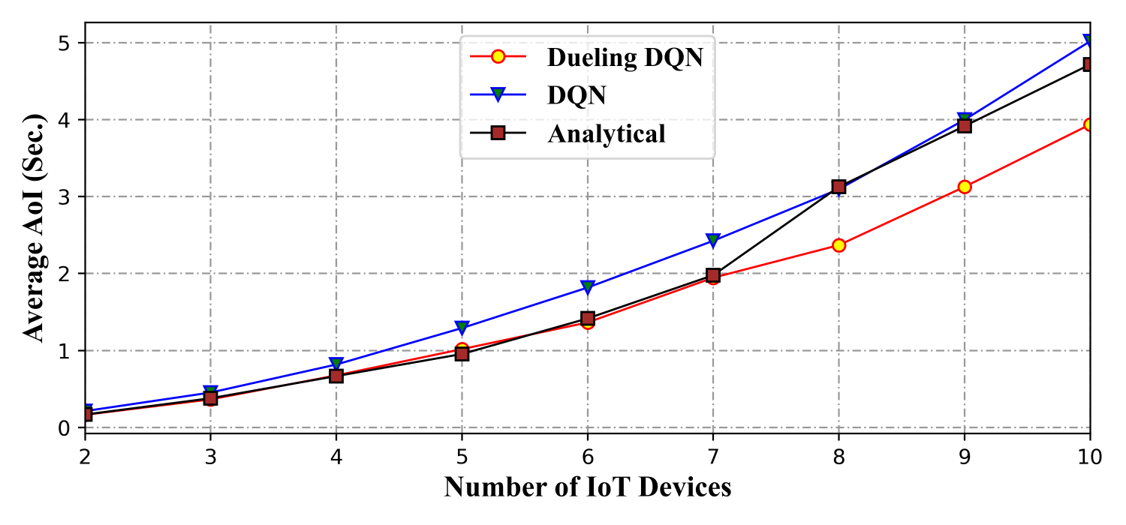}
  \caption{Performance of average end-to-end AoI with the increasing number of IoT devices.}\label{fig: DAOI}
\end{figure}

Figure \ref{fig: DAOITime} evaluates the average end-to-end AoI performance of our proposed dueling-DQN algorithm concerning different time slots and compares it with DQN and analytical results. As the number of time slots increases, the average end-to-end AoI grows monotonically concerning time $T$ for all algorithms. This is because, with more time slots, there is a higher likelihood that each packet generated by the IoT devices will take more time to schedule in the transmission channel, while also reducing the time available for generating updates. As shown in Figure \ref{fig: DAOITime}, for time slots ranging from 2 to 6, all algorithms exhibit similar performance. However, after the time slots increase beyond 6, dueling-DQN outperforms the other algorithms. This is attributed to the fact that the dueling-DQN algorithm remains more stable even with increased time slots. Both DRL algorithms (dueling-DQN and DQN) outperform the analytical results when the time slots exceed 9. This is because the DRL algorithms fully consider the effect of time slots on decision-making and utilize an experience replay buffer, which is uniformly sampled. This approach increases data efficiency and reduces variance, as the samples used in the update are less likely to be correlated.

\begin{figure}[!t]
  \centering
  \includegraphics[width=0.45\textwidth]{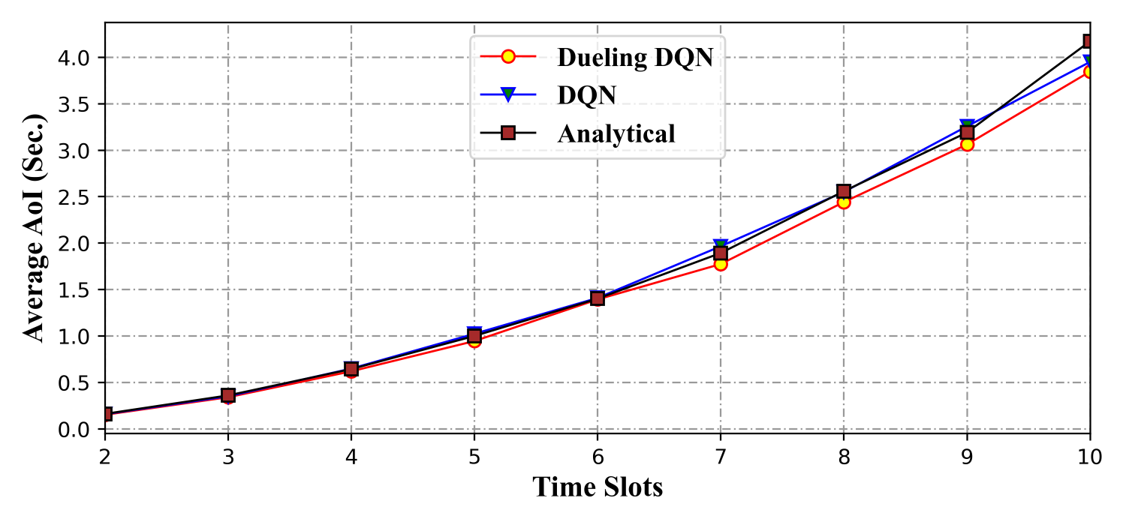}
  \caption{Performance of average end-to-end AoI with different time slots.}\label{fig: DAOITime}
\end{figure}

Figure \ref{fig: DAOIpacket} shows the performance of the algorithms when varying the packet arrival size. As the offloaded packet size increases, the average end-to-end AoI also increases. This is because it requires more time to transmit one packet for each data stream in the transmission channel, packet processing queue, and packet transmission queue, resulting in a higher average end-to-end AoI. The dueling-DQN algorithm achieves better performance in terms of AoI compared to the DQN algorithm and the analytical results. This is attributed to its ability to learn value functions without considering actions in each state, allowing it to offload more fresh information and optimize the average end-to-end AoI.

\begin{figure}[!t]
  \centering
  \includegraphics[width=0.45\textwidth]{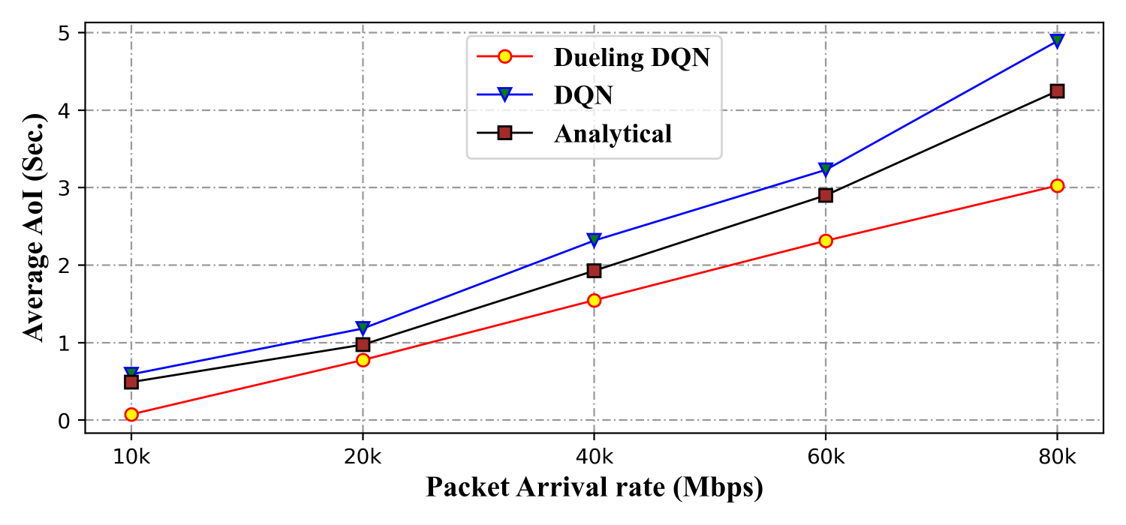}
  \caption{Performance of average end-to-end AoI with a different number of packet arrival.}\label{fig: DAOIpacket}
\end{figure}

\begin{figure}[!t]
  \centering
  \includegraphics[width=0.45\textwidth]{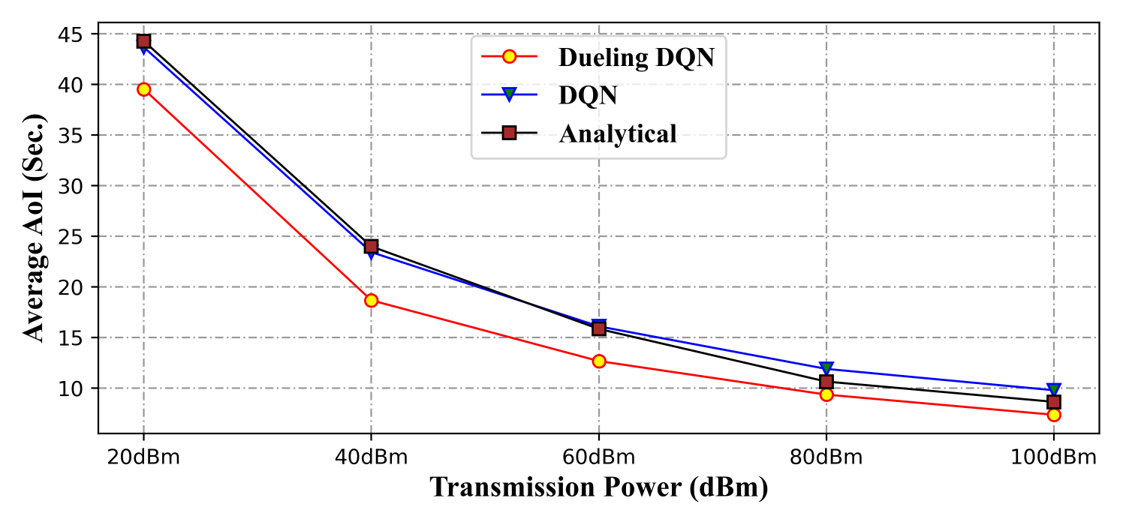}
  \caption{Performance of average end-to-end AoI with an increased number of transmit power.}\label{fig: DAOITPoer}
\end{figure}

Figure \ref{fig: DAOITPoer} illustrates the system performance evaluation with an increasing number of transmissions power. For all algorithms, the average end-to-end AoI decreases as transmission power increases. This is because higher transmission power makes it easier to offload the updated packet requirements, resulting in faster service times and a lower average end-to-end AoI. Notably, our proposed dueling-DQN algorithm performs the best, while the DQN algorithm has the worst performance. This improvement in the dueling-DQN algorithm is due to increased network stability and convergence speed, significantly reducing the average end-to-end AoI with higher transmission power. Additionally, for transmission power above 60dBm, the analytical results outperformed the DQN algorithm. Because the DQN algorithm sometimes produces wrong decisions due to overestimation and network instability, resulting in higher AoIs.

\section{Conclusion}
\label{Conclusion}
This research work considers an edge-enabled vehicular fog system and aims to minimize the average end-to-end system AoI. A time-step-based dynamic status update optimization problem is developed and analyzed, and a modified dueling-DQN-based algorithm is used to solve this problem. The simulation results demonstrate that our proposed dueling-DQN algorithm converges faster than the DQN algorithm. Furthermore, the proposed algorithm achieves a lower average end-to-end AoI compared to both the DQN algorithm and the analytical results. This improvement is attributed to the dueling-DQN algorithm's ability to learn the dynamics of the environment and make wise decisions.

Our future research will focus on leveraging the correlation of IoT device status updates to improve scheduling and offloading strategies. Although this paper assumes interoperability in vehicular fog environments, real-world implementation faces challenges due to resource and time limitations, making real-world validation a key research extension.

\bibliographystyle{IEEEtran}

\bibliography{AoI}

\clearpage

\begin{appendices}
\section{Proof of Theorem 1}
According to Equation (\ref{eq:CM1}), the time required for a packet to be transmitted from the $j$th IoT device to the edge server is given as
\begin{equation}
\begin{aligned}\label{eq:CM23}
I_{S_j}^T=\frac{S_j}{B^{I \rightarrow A}}=\frac{S_j}{W
\log _2\left(1+\frac{P_j C d_j^{I A^{-l}}}{N^2}\right)}\\
=\frac{S_j\ln 2}{W} \frac{1}{\ln \left(1+\frac{P_j C d_j^{I A}}{N^2}\right)} , 
\end{aligned}
\end{equation}
where $S \in\left\{S_1, S_2, \cdots, S_J\right\}$, denotes the size of the original packets. Since $I_{S_j}^\text{T}$ is random due to the random gain of the channel and it monotonically decreases with channel gain $C$, using Equation (\ref{eq:CM23}), can be calculated as
\begin{subequations}
\begin{equation}
\begin{aligned}\label{eq:CM24a}
\frac{1}{I_{S_j}^T} \frac{S_{j} \ln 2}{W}=\ln \left(1+\frac{P_j C d_j^{I A^{-l}}}{N^2}\right),
\end{aligned}
\end{equation}
\begin{equation}
\begin{aligned}\label{eq:CM24b}
\frac{P_j C d_j^{I A}-l}{N^2}=\left(\exp \left(\frac{1}{I_{S_j}^T} \frac{S_{j} \ln 2}{W}\right)-1\right) ,
\end{aligned}
\end{equation}
\begin{equation}
\begin{aligned}\label{eq:CM24c}
C=\frac{N^2 \left(\exp \left(\frac{S_{j} \ln 2}{W} \frac{1}{I_{S_j}^T}\right)-1\right)}{P_j d_j^{I A}-l}.
\end{aligned}
\end{equation}
\end{subequations}
We can denote that $f(I_{S_j}^T )$ is the function inversely mapping from $I_{S_j}^T$ to $C$. Which is 
$f\left(I_{S_j}^T\right)$ which is given in Equation (\ref{eq:CM24c}).
In consequence, we obtain the probability density function Equation (\ref{eq:CM26}) and cumulative distribution function Equation (\ref{eq:CM25}) of $I_{S_j}^T$, as
\begin{equation}
\begin{aligned}\label{eq:CM25}
F_{I_{S_j}^T}(x) & =P\left(I_{S_j}^T \leq x\right)=\int_{f(t)}^{\infty} \exp (-x) d x \\
= & \int_{f(t)}^{\infty} \exp \left(-\frac{\left(\exp \left(\frac{s_{j} \ln 2}{W} \frac{1}{I_{S_j}^T}\right)-1\right)}{N^{-2} P_j  d_j^{I A^{-l}}}\right) d t  \\
& =\exp \left(\frac{\left(1-\exp \left(\frac{S_{j}\ln 2}{W} \frac{1}{t}\right)\right)}{N^{-2}  P_j d_j^{I A^{-l}}}\right).
\end{aligned}
\end{equation}

\begin{equation}
\begin{aligned}\label{eq:CM26}
& f_{I_{S_j}^T}(t)=\left(-F_{I_{S_j}^T}(t)\right) \frac{d}{d t}=\left(-\exp \left(\frac{\left(1-\exp \left(\frac{S_{j} \ln 2}{W} \frac{1}{t}\right)\right)}{N^{-2} P_j d_j^{I A^{-l}}}\right)\right) \frac{d}{d t} \\
&=\frac{S_{j} \ln 2}{W N^{-2} P_j  d_j^{I A}} \frac{\exp \left(\frac{S_{j} \ln 2}{W t}+\frac{\left(1-\exp \left(\frac{S \ln 2}{W  t}\right)\right)}{N^{-2} P_j d_j^{I A}-l}\right)}{t^2}.
\end{aligned}
\end{equation}
Consequently, for packets generated by $j$th IoT device, the expected transmission time is as
\begin{equation}
\begin{aligned}\label{eq:CM27}
\resizebox{0.3\textwidth}{!}{$\mathbb{E}\left[I_j^T\right]=\mathbb{E}\left[I_{S_j}^T \mid S=S_j\right]=\int_0^{\infty} t f_{I_{S_j}^T \mid s=s_j}(t) d t$}.
\end{aligned}
\end{equation}
Finally, by substituting Equation (\ref{eq:CM26}) in Equation (\ref{eq:CM27}), we can conclude Theorem \ref{Th_1} given in Equation (\ref{eq:CM10}).

\section{Proof of Theorem 2}
The proof for this Theorem \ref{Th_2} is almost similar to Theorem \ref{Th_1}, the only difference is the size of the packets (the size of the original packets will reduce after processing on the edge server). According to the transmission rate given in Equation (\ref{eq:CM2}), the expected transmission time of packet from the edge server to the vehicular fog is given as
\begin{equation}
\begin{gathered}
\begin{aligned}\label{eq:CM28}
V_{Y_j}^T=\frac{Y_j}{B^{I \rightarrow A}}
=\frac{Y_j \ln 2}{W} \frac{1}{\ln \left(1+\frac{P_A C d^{A V^{-l}}}{N^2}\right)},
\end{aligned}
\end{gathered}
\end{equation}
where $Y \in\left\{Y_1, Y_2, \cdots, Y_J\right\}$ denotes the size of the processed packet. The rest of this proof is similar to the proof of Theorem \ref{Th_1}, and the expression is given as
\begin{subequations}
\begin{equation}
\begin{aligned}\label{eq:CM29a}
\frac{1}{V_{Y_j}^T} \frac{Y_j \ln 2}{W}=\ln \left(1+\frac{P_A C d^{A V}-l}{N^2}\right) \\
\end{aligned}
\end{equation}
\begin{equation}
\begin{aligned}\label{eq:CM29b}
C=\frac{N^2 \left(\exp \left(\frac{Y_j \ln 2}{W} \frac{1}{V_{Y_j}^T}\right)-1\right)}{P_A d^{A V^{-l}}}=f\left(V_{Y_j}^T\right),
\end{aligned}
\end{equation}
\end{subequations}
where $f(V_{Y_j}^T )$ is a function inversely mapping from $V_{Y_j}^T$ to $C$ given in Equation (\ref{eq:CM29b}), which is given as
$f\left(V_{Y_j}^T\right)$ in Equation (\ref{eq:CM29b}).
Thus, we obtain the cumulative distribution function Equation (\ref{eq:CM31}) and probability density function Equation (\ref{eq:CM30}) for $V_{Y_j}^T$, respectively
\begin{equation}
\begin{gathered}
\begin{aligned}\label{eq:CM30}
F_{V_{Y_j}^T}(t)
=\exp \left(\frac{N^2 \left(1-\exp \left(\frac{Y_{j} \ln 2}{W} \frac{1}{t}\right)\right)}{P_A d^{A V^{-l}}}\right).
\end{aligned}
\end{gathered}
\end{equation}

\begin{equation}
\begin{gathered}
\begin{aligned}\label{eq:CM31}
f_{V_{Y_j}^T}(t)
=\frac{Y_{j} \ln 2 N^2}{W P_A d^{A V}-l} \frac{\exp \left(\frac{Y_{j} \ln 2}{W t}+\frac{N^2 \times\left(1-\exp \left(\frac{Y_{j} \ln 2}{W t}\right)\right)}{P_A d^{A V}-l}\right)}{t^2}.
\end{aligned}
\end{gathered}
\end{equation}
As such, for packets transmitted from the edge node, It is possible to estimate transmission time as
\begin{equation}
\begin{aligned}\label{eq:CM32}
\mathbb{E}\left[V_j^T\right]=\mathbb{E}\left[V_{Y_j}^T \mid Y=Y_j\right]=\int_0^{\infty} t f_{V_{Y_j}^T \mid Y=Y_j}(t) dt .
\end{aligned}
\end{equation}
Finally, by substituting Equation (\ref{eq:CM31}) in Equation (\ref{eq:CM32}), we can conclude Theorem \ref{Th_2} given in Equation (\ref{eq:CM11}).

\section{Proof of Theorem 3}
Let $\mathbb{E}\left[N_j^P\right]$, denote the average number of update packets within the processing queue, and $\mathbb{E}\left[A_R^P\right]$, remaining processing time in service.
Then for $j=1$, the expectation of $\mathbb{E}\left[L_j^P\right]$, is given as
\begin{equation}
\begin{aligned}\label{eq:CM33}
& P r_{\mathrm{B}} \times \mathbb{E}\left[A_R^P\right]+\mathbb{E}\left[N_1^P\right] \times \mathbb{E}\left[A_1^P\right] \\
&=P r_{\mathrm{B}} \times \mathbb{E}\left[A_R^P\right]+\lambda_1 \mathbb{E}\left[L_1^P\right] \times \mathbb{E}\left[A_1^P\right].
\end{aligned}
\end{equation}
From Equation (\ref{eq:CM33}), we found:
\begin{equation}
\begin{aligned}\label{eq:CM34}
\mathbb{E}\left[L_1^P\right] & =\frac{P r_{\mathrm{B}} \times \mathbb{E}\left[A_R^P\right]}{1-\lambda_1 \times \mathbb{E}\left[A_1^P\right]}=\frac{P r_{\mathrm{B}} \times \mathbb{E}\left[A_R^P\right]}{1-\rho_1},
\end{aligned}
\end{equation}
where $\operatorname{Pr}_{\mathrm{B}}=\sum_{j=1}^J \lambda_j \mathbb{E}\left[A_j^P\right]=\sum_{j=1}^J \lambda_j \frac{1}{\mu^A}$, and $\rho_1$ it represents the load on the processing subsystem caused by the arrival of packets from the IoT device 1. Similarly, for $j \geq 2$, we have
\begin{equation}
\begin{aligned}\label{eq:CM35}
& \mathbb{E}\left[L_j \geq 2\right]=\frac{P r_{\mathrm{B}} \times \mathbb{E}\left[A_R^P\right]+\sum_{i=1}^{J-1} \lambda_i \mathbb{E}\left[L_i^P\right] \times \mathbb{E}\left[A_i^P\right]}{1-\sum_{i=1}^J \rho_i}.
\end{aligned}
\end{equation}
By substituting Equation (\ref{eq:CM34}) in Equation (\ref{eq:CM35}), we have
\begin{equation}
\begin{aligned}\label{eq:CM36}
\mathbb{E}\left[L_j \geq 2\right]
& =\frac{P r_{\mathrm{B}} \times \mathbb{E}\left[A_R^P\right]}{\left(1-\sum_{i=1}^J \rho_i\right)\left(1-\rho_i\right)}.
\end{aligned}
\end{equation}
Then, based on the above analysis, we can calculate $\mathbb{E}\left[L_j^P\right]$ as
\begin{equation}
\begin{aligned}\label{eq:CM37}
\mathbb{E}\left[L_j^P\right]=\left\{\begin{array}{l}
\frac{P r_{\mathrm{B}} \times \mathbb{E}\left[A_R^P\right]}{1-\rho_j}, \quad J=1 \\
\frac{P_{r_{\mathrm{B}} \times \mathbb{E}}\left[A_R^P\right]}{\left(1-\sum_{i=1}^J \rho_i\right)\left(1-\Sigma_{i=1}^{J-1} \rho_i\right)}, \quad J \geq 2
\end{array},\right.
\end{aligned}
\end{equation}
where $\mathbb{E}\left[A_R^P\right]$ is the remaining processing time. By applying the renewal-reward theory \cite{harchol2013performance}, we can determine time averages of multiple quantities by considering only the average of a single renewal cycle. To calculate the time-average excess, we follow the following procedure
\begin{equation}
\begin{aligned}\label{eq:CM38}
\mathbb{E}\left[A_R^P\right]=\frac{\mathbb{E}\left[\left(A^P\right)^2\right]}{2 \mathbb{E}\left[A^P\right]}=\frac{\sum_{j=1}^J \rho_j^2 / \lambda_j}{\sum_{i=1}^J \rho_i},
\end{aligned}
\end{equation}

$\text{where},\left\{\begin{array}{l}
\mathbb{E}\left[A^P\right]=\sum_{j=1}^J \frac{\lambda_j}{\sum_{i=1}^j \lambda_i} \mathbb{E}\left[A_j^P\right]=\frac{\sum_{j=1}^J \rho_j}{\sum_{i=1}^J \lambda_i},\\
\mathbb{E}\left[\left(A^P\right)^2\right]=\sum_{j=1}^J \frac{\lambda_j}{\sum_{i=1}^j \lambda_i} \mathbb{E}\left[\left(A_J^P\right)^2\right]=\frac{\sum_{j=1}^J \rho_j^2 / \lambda_j}{\sum_{i=1}^J \lambda_i}
\end{array}\right.$

Finally, by substituting Equation (\ref{eq:CM38}) in Equation (\ref{eq:CM37}), we can conclude Theorem \ref{Th_3} given in Equation (\ref{eq:CM13}).

\section{Proof of Theorem 4}
To prove Theorem \ref{Th_4}, we adopt the PME to derive an approximation of the expectation $\mathbb{E}\left[L_j^T\right]$, $\forall j \in J,$. An average waiting time for a typical packet in the transmission queue can be calculated as
\begin{equation}
\begin{aligned}\label{eq:CM39}
\mathbb{E}\left[L_{a v}^T\right]=\sum_{j=1}^J P r_j^T \mathbb{E}\left[L_j^T\right]=\frac{\Sigma_{j=1}^J \lambda_j}{\Sigma_{i=1}^J \lambda_i} \mathbb{E}\left[L_j^T\right],
\end{aligned}
\end{equation}
where $Pr_j^T$ Probability of an IoT device sending a packet to the transmission queue $j$, and $\mathbb{E}\left[L_j^T\right]$ is the estimation of its waiting time. From Theorem \ref{Th_1}, a typical packet spends the following amount of time in the transmission subsystem
\begin{equation}
\begin{aligned}\label{eq:CM40}
& \mathbb{E}\left[I_{a v}^T\right]=\sum_{j=1}^J \mathbb{E}\left[I_{S_j}^T \mid S=S_j\right] \operatorname{Pr}\left(S=S_j\right) \\
= & \sum_{j=1}^J \frac{\lambda_j}{\sum_{i=1}^J \lambda_i} \mathbb{E}\left[I_{S_j}^T \mid S=S_j\right]=\frac{\Sigma_{j=1}^J \lambda_j \mathbb{E}\left[I_j^T\right]}{\sum_{i=1}^J \lambda_i},
\end{aligned}
\end{equation}
where $\mathbb{E}\left[I_j^T\right]$ the estimation time in (\ref{eq:CM10}). Then, by Little’s law \cite{harchol2013performance} applying Equation (\ref{eq:CM39}) to the transmission subsystem, obtain as
\begin{equation}
\begin{gathered}\label{eq:CM41}
\mathbb{E}\left[L_{a v}^T\right]=\frac{\mathbb{E}\left[X^T\right]}{\sum_{i=1}^J \lambda_i}-\mathbb{E}\left[I_{a v}^T\right] \\
=\operatorname{Pr}_j^T \mathbb{E}\left[L_j^T\right]=\frac{\sum_{j=1}^J \lambda_j}{\sum_{ij=1}^J \lambda_i} \mathbb{E}\left[L_j^T\right]=\frac{\mathbb{E}\left[x^T\right]}{\sum_{i=1}^J \lambda_i}-\mathbb{E}\left[I_{a v}^T\right],
\end{gathered}
\end{equation}
where $\mathbb{E}\left[X^T\right]$ estimated total number of packets transmitted by the transmission system, $\mathbb{E}\left[I_{a v}^T\right]$ is given by Equation (\ref{eq:CM40}), as $X^{T}$ Since we are dealing with a random integer value, the PME can be used to represent its probability mass function, which is a uniquely correct and self-consistent method of estimating probability distributions \cite{shore1982information}. Let $Pr(X^T)$ be the probability that the number of customers served in a busy period $X$, and $M$ be the moments of $Pr(X^T)$. In general, $M$ can be expressed as a function of $\lambda_{j}$ and the service time. Then the maximum entropy distribution is
\begin{equation}
\begin{aligned}\label{eq:CM42}
& \operatorname{Pr}\left(X^T\right)=Z^{-1} \exp \left(-\sum_{m=1}^M \alpha_m(x)^m\right) \\
& \quad=Z^{-1} \exp \left(\alpha_1 x\right), \forall x \in\{0,1,2, \cdots,\},
\end{aligned}
\end{equation}
\begin{equation}
\begin{aligned}\label{eq:CM43}
\text { where } Z=\left(-\sum_{x=1}^M \alpha_m(x)^m\right)=\left(1-\exp \left(-\alpha_1\right)\right)^{-1} ,
\end{aligned}
\end{equation}
where $\alpha_m$ is the Lagrange multiplier associated with $m$th moment of the random variable $X^T$, By using Little’s law to the transmission queue and merging Equations (\ref{eq:CM42}) and (\ref{eq:CM43}), we have the following equations
\begin{equation}
\begin{aligned}\label{eq:CM44}
\operatorname{Pr}\left(X^T=0\right)=1-\rho_T \approx\left(1-\exp \left(-\alpha_1\right)\right),
\end{aligned}
\end{equation}
where $\rho_T=\sum_{j=1}^J \lambda_j \mathbb{E}\left[I_j^T\right] \quad$ probability of a busy server. As such, using this analysis, we have
\begin{equation}
\begin{aligned}\label{eq:CM45}
\mathbb{E}\left[X^T\right] \approx \frac{\partial \ln \left(1-\exp \left(-\alpha_1\right)\right)}{\partial \alpha_1}=\frac{\sum_{j=1}^J \lambda_j \mathbb{E}\left[I_j^T\right]}{1-\Sigma_{i=1}^J \lambda_i \mathbb{E}\left[I_j^T\right]}.
\end{aligned}
\end{equation}
In addition to this packet waiting or being served in the processing queue, packets preceding it are sent sequentially into the transmission queue.
Finally, by substituting Equation (\ref{eq:CM39}) and Equation (\ref{eq:CM45}) in Equation (\ref{eq:CM41}), we can conclude Theorem \ref{Th_4} given in Equation (\ref{eq:CM14}).
\end{appendices}
\ifCLASSOPTIONcaptionsoff
  \newpage
\fi

\end{document}